\documentclass[aps,pra,reprint,twoside,nofootinbib]{revtex4-2}
\usepackage[english]{babel}

\usepackage{graphicx}
\usepackage{dcolumn}
\usepackage{bm}
\usepackage{physics}
\usepackage{amsmath, amssymb, amsfonts, amsthm, dsfont}
\usepackage{svg}
\usepackage{color}
\usepackage{url}
\usepackage[hidelinks]{hyperref}
\usepackage[capitalise]{cleveref}
\usepackage{microtype}
\usepackage[normalem]{ulem}
\usepackage{stmaryrd}
\usepackage{multirow}
\usepackage{booktabs}
\usepackage[caption=false]{subfig}
\usepackage{appendix}

\definecolor{darkblue}{RGB}{50,10,180}
\definecolor{orangef}{RGB}{210,100,20}

\newcommand{\id}{\mathds{1}}

\newtheorem{mainTheorem}{Theorem}
\newtheorem{Lemma}{Lemma}[section]
\newtheorem{mainLemma}{Lemma}
\newtheorem{Corollary}{Corollary}[section]

\newtheorem{Proposition}{Proposition}[section]
\newtheorem{mainProposition}{Proposition}

\theoremstyle{definition}
\newtheorem{Definition}{Definition}[section]
\theoremstyle{remark}
\newtheorem*{Remark}{Remark}
\theoremstyle{plain}

\newcommand{\s}{\mathsf{s}}
\renewcommand{\t}{\mathsf{t}}
\renewcommand{\a}{\mathsf{a}}
\renewcommand{\b}{\mathsf{b}}
\newcommand{\cor}[1]{\langle #1 \rangle}

\begin{document}

\title{Quantum statistics in the minimal scenario}

\author{Victor Barizien}
\affiliation{Université Paris Saclay, CEA, CNRS, Institut de physique théorique, 91191 Gif-sur-Yvette, France}
\author{Jean-Daniel Bancal}
\affiliation{Université Paris Saclay, CEA, CNRS, Institut de physique théorique, 91191 Gif-sur-Yvette, France}

\date{\today}

\begin{abstract}
In any given experimental scenario, the rules of quantum theory provide statistical distributions that the observed outcomes are expected to follow. The set formed by all these distributions contains the imprint of quantum theory, capturing some of its core properties. So far, only partial descriptions have been known for this set, even in the simplest scenarios. Here, we obtain the analytical description of a complete set of quantum statistics in terms of extremal points. This is made possible by finding all bipartite quantum states and pairs of binary measurements which can be self-tested, i.e.~identified from statistics only. Our description provides a direct insight into the properties and limitations of quantum theory. These are not expressed in terms of Hilbert spaces, but rather directly in terms of measurement observation statistics.
\end{abstract}

\maketitle


Quantum theory is surprising in many regards, one of which being its intrinsic inability 
to predict the exact results of experiments, such as the outcome that is going to be observed when a measurement is performed on a physical system. Indeed, quantum physics only foretells the statistical distribution of these outcomes and is, as such, a fundamentally probabilistic theory. Whereas one might expect this to be a limitation of the theory, it turns out that the probabilistic predictions of quantum theory often exceed classical and deterministic ones, leading to new possibilities. The puzzle that a probabilistic theory could somehow be more powerful than a deterministic one has stimulated deep questions~\cite{Einstein35}. Today, the unpredictability of quantum theory is considered a resource~\cite{Pironio10a} and in light of this new perspective, a natural question arises: what are the fundamental limits of quantum theory's probabilistic predictions?

Addressing this question requires distinguishing probabilistic predictions that can admit a quantum explanation from those which don't. Interestingly, because quantum statistics can violate Bell inequalities~\cite{Brunner14}, determining whether a given set of probability distributions is compatible with quantum theory turns out to be a highly nontrivial task. In fact, this problem amounts to inverting Born's rule~\cite{Born26}, one of quantum theory's cornerstones. This rule expresses in simple terms the probability assigned by quantum theory to the potential outcomes of a measurement as a function of a system's \emph{realization} -- including its state and measurement operators. While the statistics produced by a quantum realization is unique, a set of quantum statistics can admit many quantum realizations simultaneously, rendering their characterization particularly challenging.

Recently, several efforts have highlighted the wide-ranging consequences of determining the \emph{quantum set} -- the set of quantum statistical predictions. Beyond allowing to test whether experimental observations admit a quantum description or not, it was suggested that by touching upon the limits of quantum theory itself this description would give access to fundamental principles satisfied by the theory. Similarly to the constancy of the speed of light in relativity, such principles could open the way for a fully principle-based formulation of quantum theory~\cite{Scarani12}. Several information principle candidates have been proposed, such as no-signaling, information causality, macroscopic locality or local orthogonality~\cite{Popescu92a,Pawlowski09,Navascues10,Fritz13}. So far, none of these have succeeded in reproducing quantum predictions. Thus, although being satisfied by the theory, the principles identified until now fall short in providing a physical explanation for why quantum correlations are limited the way they are and the search for such a quantum principle remains open.

Knowledge of the quantum set also has direct implications for quantum applications. Indeed, by inverting Born's rule without any a priori on the underlying Hilbert space or on the description of the quantum system at hand, this approach provides a way of accomplishing tasks that is fundamentally trustful, in the sense that it is independent of the precise modeling of the physical devices at hand~\cite{Acin07}. This device-independent approach has been shown to present an interest in entanglement detection and quantification~\cite{Barreiro13,Moroder13} and is now a standard framework in numerous tasks. The security of adversarial protocols particularly benefits from assessments that are insensitive to the implementations details~\cite{Nadlinger22,Zhang22a,Liu22a}. Since device-independent information processing analyses rely on observed statistics, they depend directly on the characterization of the quantum set.

More generally, it has been shown that points in the quantum set are pertinent to a wide range of topics, including the study of correlations in many-body systems~\cite{Tura14a} or of quantum computing advantage~\cite{Bravyi18}. Determining the limits of this set is thus crucial to identify new possibilities and limitations in quantum information science.

Considering this problem already in the 1980s, Tsirelson was the first to obtain bounds on the quantum set~\cite{Cirelson80,Tsirelson93}. Several results followed, notably from Landau and Masanes~\cite{Landau88,Masanes03}, until a major progress was achieved by Navascués, Pironio and Acín (NPA) in the form of a hierarchy of semidefinite programming~\cite{Navascues07,Navascues08}. This hierarchy is now a central tool of quantum information science~\cite{Skrzypczyk23,Mironowicz24,Tavakoli24}, with an impact reaching optimization theory~\cite{Bhardwaj21}. Concretely, the hierarchy defines a family of problems of increasing complexity which approximate better and better the quantum set from the outside and guarantees convergence as the level of the hierarchy goes to infinity. At a fixed hierarchy level, this technique allows deriving necessary conditions for the quantum set~\cite{deVicente15}, thus excluding that some behaviors admit a quantum representation. However, since the NPA technique can generally not guarantee that specific statistics are quantum, its implications remain elusive on the boundary of the quantum set.

Recently, fresh insight was gained on the quantum set by analytical studies which showed that it admits flat nonlocal boundaries~\cite{Goh18, Chen23}, as well as pointy nonlocal extremal points~\cite{Barizien23,Barizien24}. New curved regions of the set's boundaries were also identified analytically~\cite{Liu24} and several conjectures were formulated on the boundary of the quantum set~\cite{Ishizaka18,Mikosnuszkiewicz23}.

In this work, we consider the question of identifying the limits of the quantum set from the perspective of self-testing~\cite{Mayers04}. Namely, a set of statistics, or \emph{behavior}, is said to self-test a quantum realization with state $\rho$ and measurements $\{M\}$ iff it can only be obtained through Born's rule by quantum realizations $(\rho',\{M'\})$ which are related to $(\rho,\{M\})$ by local isometries~\cite{Supic20}. Self-testing is a powerful tool to analyze quantum behaviors: it has played an important role in solving Connes' embedding problem, a major conjecture of operator algebra~\cite{Ji21}, and is a central part of several quantum protocols such as the certification of quantum devices~\cite{Magniez06,Sekatski18} or delegated quantum computing~\cite{Reichardt13,Gheorghiu15,Hajdusek15,McKague16}.

So far, numerous families of states have been shown to be self-testable through some of their behaviors~\cite{Coladangelo17,Supic18}. However, only the statistics obtained when measuring a maximally entangled state in the minimal scenario, where two parties each equipped with two binary measurements, have been fully characterized by self-testing~\cite{Wang16}. These behaviors correspond to the boundary of the quantum set in this scenario with vanishing marginals, which was described earlier by Tsirelson, Landau and Masanes, see~\cite{Le23}.

Here, we obtain all the self-tests in this scenario. This reveals previously unknown boundaries of the quantum set. In turn, it allows us to identify all extremal points and their corresponding quantum realizations, thus providing a complete description of the quantum set in the minimal scenario.

\begin{figure}
\includegraphics[width=0.45\textwidth]{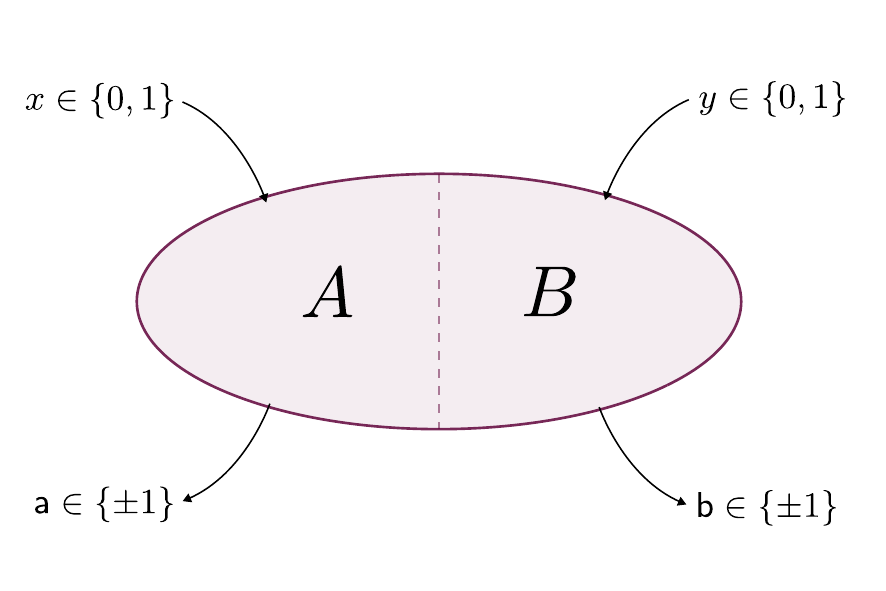}
\caption{Quantum measurements on a bipartite system produces outcomes which follow a set of probability distributions that depend on the quantum state $\rho_{AB}$ and on the measurement operators $A_x, B_y$, as dictated by Born's rule. Here, we consider the set $\mathcal{Q}$ of all probability distributions which can be obtained according to quantum theory in the celebrated Clauser-Horne-Shimony-Holt (CHSH) scenario, where both parties can choose between two possible measurement settings, i.e.~$x,y=0,1$, and where the outcomes of the measurements $\a,\b=\pm1$ are binary. This includes all measurements on every possible quantum states. By identifying all the extremal quantum behaviors, we obtain a description of quantum theory in this setting. }
\label{fig:scenario}
\end{figure}

We derive our results in the bipartite setting depicted in \cref{fig:scenario}, in which two users, Alice and Bob, each have access to a shared quantum state $\rho_{AB}\in L(\mathcal{H}_A\otimes\mathcal{H}_B)$, where $\mathcal{H}_{A/B}$ denotes their respective Hilbert space\footnote{Here, $L(\mathcal{H})$ refers to the space of linear applications on the Hilbert space $\mathcal{H}$.}. In the scenario we consider, each party can perform one of two possible measurements, denoted by $x,y=0,1$, each with two possible outcomes, denoted $\a,\b=\pm 1$, and described by the POVM elements $P_{\a,x}^A$, $P_{\b,y}^B$ respectively. To each POVM, one can associate hermitian operators, $A_x\in L(\mathcal{H}_A)$ and $B_y\in L(\mathcal{H}_B)$, with eigenvalues in $[-1,1]$, defined as $A_x = P_{1,x}^A-P_{-1,x}^A$ and $B_y = P_{1,y}^B-P_{-1,y}^B$. In this setting, a behavior is fully determined according to Born's rule by the eight real parameters
\begin{align}
\cor{A_x}&=\tr(\rho_{AB} A_x\otimes \ \id\ ),\\
\cor{B_y}&=\tr(\rho_{AB}\ \id\ \otimes B_y),\\
\cor{A_x B_y}&=\tr(\rho_{AB} A_x\otimes B_y),
\end{align}
which encode a vector, or \emph{point}, $\Vec P \in \mathbb{R}^8$. The set of all such vectors $\Vec P$ form the quantum set $\mathcal{Q}$ of interest. It should be stressed that all states and measurements in arbitrary Hilbert spaces $\mathcal{H}_A$ and $\mathcal{H}_B$, including infinite dimensional Hilbert spaces, can be considered here. Hence, this set encodes all statistics that can be achieved in this setting within the framework of quantum theory.

The fact that the Hilbert space dimension is unbounded allows simplifying the analysis of $\mathcal{Q}$ in two ways. First, it ensures that the quantum set is convex and therefore that it is fully described as the convex hull of its extremal points: $\mathcal{Q}= \text{Conv}(\text{Ext}(\mathcal{Q}))$ \cite{Rockafellar70}. We can thus focus on the description of these extremal points, which are known to be infinitely many. Second, it ensures by Naimark's dilation theorem that points of $\mathcal{Q}$ can be realized in terms of projective measurements acting on pure states. The description of extremal points can then be further simplified by Jordan's lemma, which guarantees that they can be realized on Hilbert spaces of dimension two whenever the parties hold two binary measurements, i.e.~$\text{Ext}(\mathcal{Q})\subset \mathcal{Q}_2$. A recent result by Mikos-Nuszkiewicz and Kaniewski~\cite{Mikosnuszkiewicz23} further guarantees that such realizations can be parametrized in the simple real form 
\begin{subequations}\label{eq:realization}
\begin{align}
\ket{\phi_\theta}&=\cos{\theta}\ket{00}+\sin\theta\ket{11},\\
A_x&=\cos a_x \sigma_z + \sin a_x \sigma_x,\\
B_y&=\cos b_y \sigma_z + \sin b_y \sigma_x.
\end{align}
\end{subequations}
Considering the symmetries of the set $\mathcal{Q}$ allows us to finally restrict our attention to parameters in the range
\begin{equation}\label{eq:range}
\theta \in [0,\pi),\ 0 \leq a_0 \leq b_0 \leq b_1 < \pi,\ a_0 \leq a_1 <\pi,
\end{equation}
see~\cref{sec:paramreduc} for more details on these restriction steps.

\begin{figure}
\includegraphics[width=0.45\textwidth]{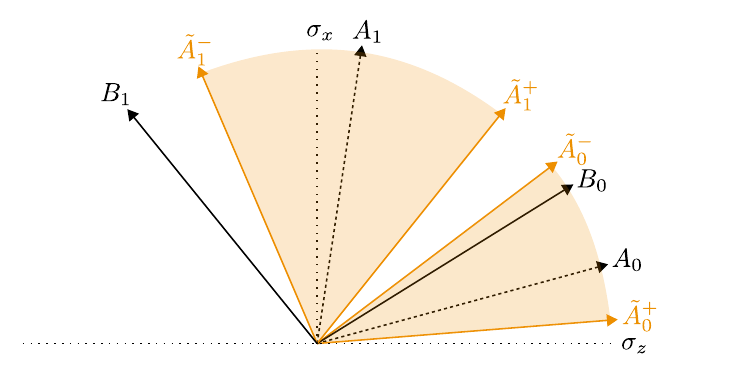}
\caption{Action of the non-linear steering map $T_\theta$ on Alice's measurements $A_0$ and $A_1$ when $\theta\in(0,\pi/4)$. Bob's measurement direction $B_0$ and $B_1$ are also represented in the same plane. The map, defined as $T_{\theta}:\ket{u}\mapsto \frac{\sqrt{2}(\braket{\phi_\theta}{u})\ket{\phi^+}}{||\braket{\phi_\theta}{u}||_2}$, describes the steering of vectors in $\mathbb{C}^2$ through the state $\ket{\phi_\theta}$. By construction, this transformation preserves the normalization so that $T_{\theta}(\ket{u}) T_{\theta}(\ket{u})^\dag$ is a unit projector for every normalized $\ket{u}$. However, the map does not preserve the scalar product: in general, $T_{\theta}(\ket{u})^\dag T_{\theta}(\ket{u'})\neq\braket{u}{u'}$. For this reason, we associate to every projective measurement of Alice $A_x=\ketbra{u_{+1,x}}{u_{+1,x}} - \ketbra{u_{-1,x}}{u_{-1,x}}$ two image measurements through this map, one for each projector: $\Tilde A^{\a}_x = \a (2T_\theta(\ket{u_{\a,x}}) T_\theta(\ket{u_{\a,x}})^\dag-\id)$. When the initial measurement $A_x$ has an angle $a_x$ with respect to $\sigma_z$ in the ($\sigma_z,\sigma_x$)-plane, the modified measurements still lie on this plane, but with angles $\Tilde{a}^{\a}_x = 2  \operatorname{atan}\left(\tan(\frac{a_x}{2})\tan(\theta)^\a \right)$ with respect to $\sigma_z$, as illustrated above. We say that Alice's and Bob's measurement \textit{fully alternate} if the measurement angles verify $0 \leq [\Tilde{a}_0^\s]_\pi \leq b_0 \leq [\Tilde{a}_1^\t]_\pi \leq b_1 < \pi$ for all $\s,\t \in \{-1,1\}$, where $[\alpha]_\pi :\equiv \alpha[\pi]$. In~\cref{sec:steering}, we show that the expectation value of these modified operators on the maximally entangled state $\ket{\phi^+}$ relate to the initial statistics as $\bra{\phi^+}\Tilde A^\a_x B_y\ket{\phi^+}=\frac{\langle A_x B_y \rangle+\a \langle B_y \rangle}{ 1+\a \langle A_x \rangle}$.}
\label{fig:nonlinearMap}
\end{figure}

Our first result concerns the set $\mathcal{Q}_2$ of statistics that can be obtained by measuring quantum systems of dimension 2. Due to the dimension restriction, this set is non-convex~\cite{Donohue15}, which makes it hard to characterize. Defining the correlators $[\Tilde A_x^\a B_y] = \frac{\langle A_x B_y \rangle+\a \langle B_y \rangle}{ 1+\a \langle A_x \rangle}$  for all $\a\in\{\pm 1\}$, $x, y\in\{0,1\}$, we have the following necessary condition for pure realizations in $\mathcal{Q}_2$.

\begin{mainProposition}[Necessary condition for pure projective realizations of local dimension 2] \label{prop:loc2}
Any quantum point obtained from projective measurements on a pure entangled two-qubit state verifies 
\footnotesize
    \begin{equation} \label{eq:ineqmargin}
\begin{array}{r}
-\pi \leq-\operatorname{asin} [\Tilde A_0^\s B_0] +\operatorname{asin} [\Tilde A_1^\t B_0] +\operatorname{asin} [\Tilde A_0^\s B_1]+\operatorname{asin}[\Tilde A_1^\t B_1]\leq \pi \\
-\pi \leq \operatorname{asin} [\Tilde A_0^\s B_0] -\operatorname{asin} [\Tilde A_1^\t B_0] +\operatorname{asin} [\Tilde A_0^\s B_1]+\operatorname{asin}[\Tilde A_1^\t B_1] \leq \pi \\
-\pi \leq \operatorname{asin} [\Tilde A_0^\s B_0] +\operatorname{asin} [\Tilde A_1^\t B_0] -\operatorname{asin} [\Tilde A_0^\s B_1]+\operatorname{asin}[\Tilde A_1^\t B_1]\leq \pi \\
-\pi \leq \operatorname{asin} [\Tilde A_0^\s B_0] +\operatorname{asin} [\Tilde A_1^\t B_0] +\operatorname{asin} [\Tilde A_0^\s B_1]-\operatorname{asin}[\Tilde A_1^\t B_1]\leq \pi
\end{array}
\end{equation}
\normalsize
for all $\s,\t \in \{\pm1\}$.
\end{mainProposition}
The main idea of the proof consists in considering the quantum state $\ket{\phi_\theta}$ as a non-linear map $T_\theta$ that brings an arbitrary pure entangled qubit behavior in relation with 4 distributions issued from the maximally entangled state, see \cref{fig:nonlinearMap}. Constraints on maximally entangled statistics derived by Masanes~\cite{Masanes03} then imply the necessary constraints on the original statistics issued from a non-maximally entangled state that are given here, see~\cref{sec:proofprop1} for more details.

Note that the correlators $[\Tilde A_x^\a B_y]$ are well-defined, except in the case where some of Alice's marginals are equal to~$\pm 1$. This can only happen when $\theta\equiv 0 [\pi/2]$, which gives rise to local points. Note also that the labeling of the subsystems is arbitrary here, so additional necessary conditions can be obtained by exchanging Alice and Bob.

One easily verifies that when the marginal of both parties are zero, the inequalities above reduce to the 8 Masanes inequalities. These inequalities are convex and identify the boundary of the quantum set in this restricted CHSH scenario \cite{Le23}. In the general case, however, the inequalities \cref{eq:ineqmargin} do not describe a convex set. Therefore, they cannot describe the set $\mathcal{Q}$. Indeed, measurements on a high-dimensional state or on a mixed two-qubit state can result in quantum points that do not satisfy the above inequalities.

It has been shown that the saturation of the Masanes inequalities can only be achieved by self-testing behaviors, which identify a singlet state and real unitary measurements \cite{Wang16}. It is thus natural to ask whether the points saturating some of the above conditions could lead to self-testing as well. It  turns out that saturating a single inequality is not sufficient. Similarly, one can find non-extremal points saturating two inequalities. 
However, whenever three of the conditions are met for different values of $(s,t)$, the fourth one is also. Indeed, out of the above 32 conditions, only 24 are linearly independent. In this special case, we show that the resulting statistics self-test a qubit realization.

\begin{mainLemma} \label{thm:self-test} 
Any nonlocal behavior which satisfies
\begin{equation}\label{eq:conditions}
        \scalebox{0.9}{$\operatorname{asin} [\Tilde A_0^\s B_0] +\operatorname{asin} [\Tilde A_1^\t B_0] - \operatorname{asin} [\Tilde A_0^\s B_1]+\operatorname{asin}[\Tilde A_1^\t B_1] =\pi$}
\end{equation}
for all four $(\s,\t) \in \{\pm 1\}^2$, self-tests a quantum realization of local dimension 2. 
\end{mainLemma}
\noindent In particular, all realizations in the range~\cref{eq:range} for which the measurement angles \textit{fully alternates}, i.e.~verify  
\begin{equation}\label{eq:fullAlt}
        0 \leq [\Tilde{a}_0^\s]_\pi \leq b_0 \leq [\Tilde{a}_1^\t]_\pi \leq b_1 < \pi
\end{equation}
satisfy, up to relabelling, the condition of~\cref{thm:self-test}. As such, those realizations are self-tested by their associated quantum point, which is therefore extremal in~$\mathcal{Q}$~\cite{Goh18}. Note that conversely, satisfying \cref{eq:conditions} for all $(\s,\t)\in\{\pm 1\}^2$ self-tests realizations that fully alternate (up to relabeling).

The proof of \cref{thm:self-test}, given in~\cref{sec:proofLemma1}, strongly relies on the steering transformation. This time, the nonlinear transformation is used to map self-testing statements on a maximally entangled state into a self-testing statement on a partially entangled one. This is made possible by strong geometrical relations between the vectors $\Tilde A_x^{\a} \ket{\phi^+}$ and $B_y \ket{\phi^+}$ issued from the steered realization when \cref{eq:conditions} holds. 

Note that the condition of nonlocal behavior implies that the marginals cannot take the value $\pm1$, which guarantees that the correlators $[\Tilde A_x^\s B_y]$ are well-defined. Indeed, any vector with a single marginal probability equal to $\pm 1$ would have more than two zeros of probabilities on the same line or column of the probability table, ensuring that the quantum point is local~\cite{Chen23}.

\cref{thm:self-test} identifies many extremal points of the quantum set and provides self-testings for all qubit partially-entangled states with a large family of measurement settings. However, it does not say whether points which do not satisfy the equality case are extremal in $\mathcal{Q}$. In order to address this, we consider the complementary question of showing that realizations which do not satisfy the full alternation condition \cref{eq:fullAlt} lead to behaviors that are non-extremal. As a first step in this direction, we prove that in the absence of full alternation, the statistics are non-exposed, i.e.~that they do not uniquely reach the Tsirelson bound of any Bell expression. Note that since non-exposed points include both decomposable points and some extremal points, this result in itself does not allow concluding yet.

\begin{mainLemma} \label{thm:non-exp}
    Consider a quantum realization with state $\ket{\phi_\theta}$ and measurements satisfying \cref{eq:range}. If the following series of inequalities does not hold:
    \begin{equation}
        \forall \s, \t, \ 0 \leq [\Tilde{a}_0^\s]_\pi < b_0 < [\Tilde{a}_1^\t]_\pi < b_1 < \pi,
    \end{equation}
    where $\Tilde{a}^{\a}_x =  \operatorname{atan}(\tan(\frac{a_x}{2})\tan(\theta)^\a)$, and $[\alpha]_\pi:\equiv\alpha[\pi]$, then the corresponding point is non-exposed in $\mathcal{Q}$. 
\end{mainLemma}

The idea that we employ in~\cref{sec:proofLemma2} to prove that these points are not exposed is to consider all possible Bell expressions, and to show for each of them that the value achieved by the considered quantum point can be also achieved by a distinct quantum point. In order to prove this, we first find necessary conditions that Bell expressions maximized by the considered statistics must satisfy. This allows reducing the dimension of the search space for Bell expressions from 8 to 3. For all remaining Bell expressions, we then provide a local point which gives the same Bell value as the considered quantum point.

Equipped with Lemmas 1 and 2, let us consider a qubit realization of the form \cref{eq:realization,eq:range}. If the measurement angles satisfy the alternating condition \cref{eq:fullAlt}, then \cref{thm:self-test} ensures that the quantum realization is self-tested and therefore the point is extremal. This shows that the alternating property \cref{eq:fullAlt} is constitutive of these points: no point obtained with non-alternating settings can be obtained with other, alternating settings. Using this property together with \cref{thm:non-exp} ensures that all realizations that do not verify \cref{eq:fullAlt} correspond to points on the interior of the set of non-exposed points of~$\mathcal{Q}$. However, Straszewicz's theorem  states that all extremal points are limits of exposed points \cite{Rockafellar70}. This ensures that non-alternating realizations lead to non-extremal points of~$\mathcal{Q}$. This can be summarized in the following theorem.

\begin{mainTheorem}[Characterization of $\text{Ext}(\mathcal{Q})$] \label{thm:charac}
\
    \begin{enumerate}
    \item
    A nonlocal point in the CHSH scenario is extremal in $\mathcal{Q}$ iff
    \begin{equation}\label{eq:thm1eq1}
        \forall \, \mathsf{u} \in \{\pm 1\}^2, \sum_{x,y} \epsilon_{xy} \asin [\Tilde A_x^{\mathsf{u}_x} B_y] = \pi
    \end{equation}
    for some $\epsilon_{xy} \in \{\pm 1\}$ such that $\prod_{x,y} \epsilon_{xy} = -1$.

    \item
    A quantum realization leads to a nonlocal extremal point iff it can be mapped by local channels and relabelings to a quantum realization on the entangled state $\ket{\phi_\theta}$ with measurements satisfying \cref{eq:range} s.t.
    \begin{equation} \label{eq:thm1eq2} 
        \forall (\s,\t) \in \{\pm 1\}^2, 0 \leq [\Tilde{a}_0^\s]_\pi \leq b_0 \leq [\Tilde{a}_1^\t]_\pi \leq b_1 < \pi.
    \end{equation}
    \end{enumerate}
\end{mainTheorem}

\begin{figure*}

\subfloat[\label{fig:sub1}]{%
  \includegraphics[width=0.49\linewidth]{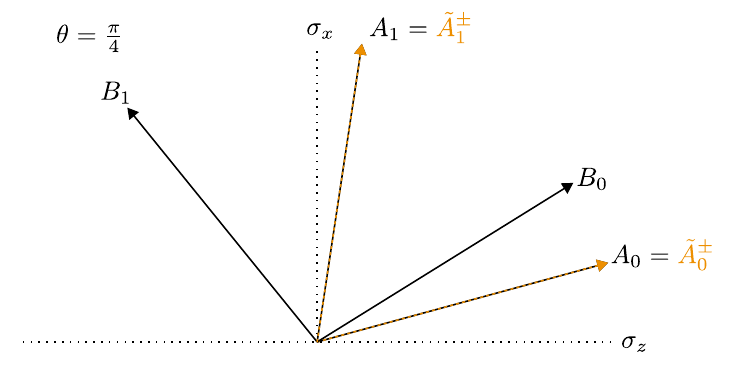}%
}\hspace*{\fill}%
\subfloat[\label{fig:sub2}]{%
  \includegraphics[width=0.49\linewidth]{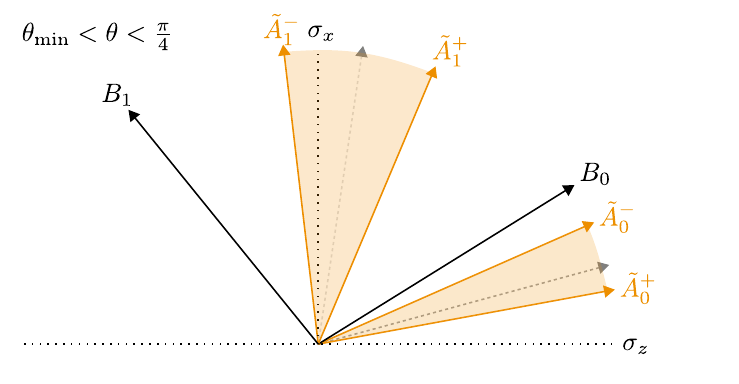}%
}

\medskip

\subfloat[\label{fig:sub3}]{%
  \includegraphics[width=0.49\linewidth]{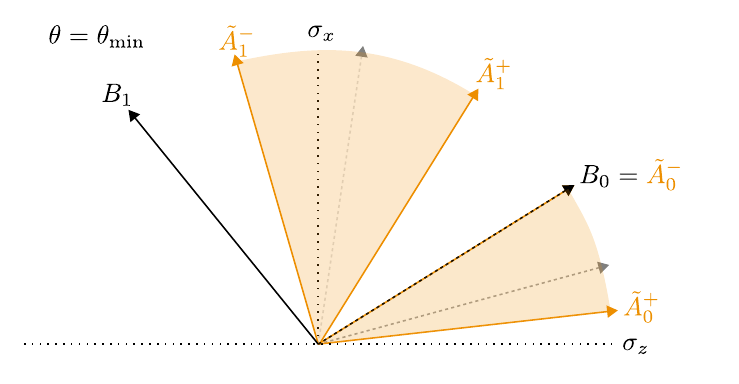}%
}\hspace*{\fill}%
\subfloat[\label{fig:sub4}]{%
  \includegraphics[width=0.49\linewidth]{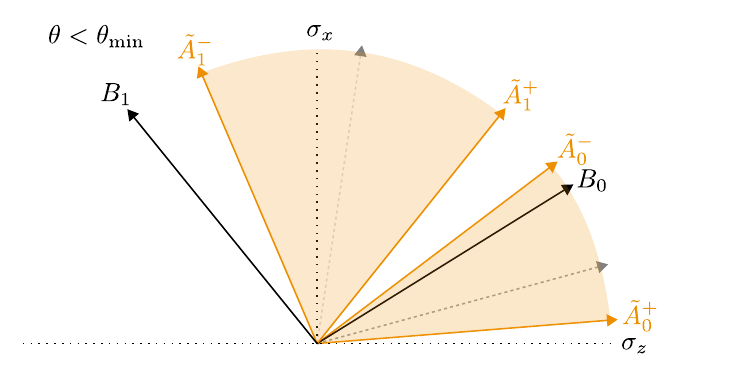}%
}

\caption{Geometrical representation of the realizations leading to extremal points of the quantum set in the CHSH scenario. Up to local isometries and relabelings, these realizations correspond to measurements on a partially entangled two-qubit state $\ket{\phi_\theta}$ taking place in the $(\sigma_z,\sigma_x)$-plane and verifying the fully alternating condition~\cref{eq:fullAlt}. If the initial measurement angles alternate, i.e.~verify $0\leq a_0 \leq b_0 \leq a_1 \leq b_1 < \pi$, then the corresponding points are extremal for all entanglement parameters $\theta$ in $[\theta_{\text{min}},\pi/4]$, where the threshold entanglement value is given by $\theta_{\min}=\max \{\theta\leq \pi/4 \text{ s.t. } \Tilde{a}_x^\a = b_y\}$. As such, for $\theta=\pi/4$, pictured in~\cref{fig:sub1}, the point is extremal. For $\theta_{\min} < \theta < \pi/4$, pictured in~\cref{fig:sub2}, the orange areas and black arrows still alternate, so the point is extremal. For $\theta=\theta_{\min}$, see~\cref{fig:sub3}, the point is still extremal but non-exposed in $\mathcal{Q}$. Last, for $\theta<\theta_{\min}$, like in~\cref{fig:sub4}, the black arrow of $B_0$ is inside an orange area: the fully alternating property is broken, and the point admits a decomposition.}
\label{fig:angles}

\end{figure*}

Note that the full alternance condition is equivalent whether the steering transformation is applied on Alice's measurements (as done here) or on Bob's ones, see \cref{sec:AliceVSBob}. Therefore, exchanging the role of Alice and Bob in conditions \cref{eq:thm1eq1,eq:thm1eq2}, with modified angles $\tilde b_y^\b$ and correlators $[A_x \tilde B_y^\b]$, identifies the same extremal point.

Together with the 8 deterministic strategies~\cite{Brunner14}, which are local extremal points of $\mathcal{Q}$, \cref{thm:charac} gives a complete characterization of the extremal point of the quantum set in the minimal scenario. It also describes the quantum realizations behind them in terms of measurements on two-qubit states, as illustrated in \cref{fig:angles}. This description implies that all nonlocal extremal points in this scenario self-test a two-qubit realization.

Note that the analytical description of the quantum set given by \cref{thm:charac} provides new insight on its geometry. Indeed, it implies that the 8-dimensional convex set is generated by a 5-dimensional sub-manifold of extremal points. Indeed, as discussed above, whenever Alice's marginals are not zero, three out of the four equations above are linearly independent. Furthermore, by \cref{thm:non-exp}, extremal points admitting a probability equal to zero verify $[\Tilde A_x^{\a} B_y]=\pm 1$ for some $\a,x,y$ and are thus non-exposed.

\acknowledgements
We thank Valerio Scarani for feedback on the manuscript.

\bibliography{bilbi}

\newpage
\onecolumngrid
\appendix 

\section*{SUPPLEMENTAL MATERIAL}

In this supplemental material, we explain the technical details underlying the results presented in the paper. In particular, \cref{sec:paramreduc} focuses on the parametrization of the extremal points in the CHSH scenario, \cref{sec:steering} on the steering transformation $T_\theta$, while \cref{sec:proofLemma1,sec:proofLemma2} provide detailed proofs of Lemma 1 and 2 respectively. 

\section{Sufficient parametrization for the extremal points of the quantum set in the CHSH scenario} \label{sec:paramreduc}

Quantum correlations are defined as the statistics that can be observed upon performing arbitrary measurements on an arbitrary quantum state. As such, describing the quantum set exhaustively requires in principle considering all possible realizations (states and measurements) in every Hilbert space. Thankfully, it is known that several realizations give rise to identical statistics, and therefore a subset of all quantum realizations is sufficient to reproduce all quantum correlations in a given Bell scenario. In turn, realizations can be even further restricted when considering extremal quantum correlations only. In this appendix, we discuss these restrictions in detail. This leads us to the explicit parametrization given in \cref{prop:qubitrealization}, which is sufficient to reproduce all extremal points in the CHSH scenario.

This appendix is organised as follows. First, we discuss up to \cref{cor:JordanBell} the way that restrictions can be made on measurement operators without affecting the set of generated statistics. Then, we incorporate restrictions on the state as well as those due to discrete symmetries of the Bell scenario.

While we focus on the case of the CHSH scenario, several of the results described here also apply to general Bell scenarios. Also, most of the steps described in this appendix were already discussed in earlier works, see e.g.~\cite{Masanes05,Mikosnuszkiewicz23}, however some of our proofs are more concise. We include a proof of every restriction step for completeness.

\begin{Definition}
    A \emph{binary measurement} -- a quantum measurement with two possible outcomes -- is described by an operator $A\in L(\mathcal{H})$ such that $A^\dagger = A$ and $-\id \preceq A \preceq \id$. Furthermore, the measurement is \emph{projective} iff $A^2 = \id$, in which case the only possible eigenvalues of $A$ are $\pm 1$ and $A$ is unitary, i.e.~$A^\dag A=A A^\dag=\id$.
\end{Definition}

\begin{Definition}
    \emph{Quantum correlations in the CHSH scenario} are the correlation vectors $\Vec P \in \mathbb{R}^8$ with components 
\begin{align}
\cor{A_x}&=\tr(\rho_{AB} A_x\otimes \ \id\ ),\\
\cor{B_y}&=\tr(\rho_{AB} \ \id\ \otimes B_y),\\
\cor{A_x B_y}&=\tr(\rho_{AB} A_x\otimes B_y),
\end{align}
    where $\rho_{AB}\in L(\mathcal{H}_A\otimes \mathcal{H}_B)$ is a quantum state (satisfying $\rho_{AB}\succeq 0$, $\Tr(\rho_{AB})=1$) and $A_x\in L(\mathcal{H}_A)$, $B_y\in L(\mathcal{H}_B)$ with $x,y\in{0,1}$ are binary measurement operators. Each such vector encodes the statistics corresponding to measuring the quantum state $\rho_{AB}$ with the corresponding two pairs of binary measurements.
\end{Definition}

\begin{Remark}
    It is convenient to arrange the components of $\Vec P$ in a table with elements indexed by $x,y\in\{-1,0,1\}$. Namely, defining $A_{-1} = \id$, $B_{-1} = \id$, we can write
    \begin{equation}
        \{\Vec P\}_{xy} = \Tr(\rho_{AB}(A_x\otimes B_y)),
    \end{equation}
    where the first element of this table, $\{\Vec P\}_{-1,-1}=\Tr(\rho_{AB})=1$, is a constant.
\end{Remark}

\begin{Remark}
This definition is equivalent to the definition of quantum correlations in terms of full probabilities, which identifies distributions of the form $p(\a\b|xy)=\Tr(\rho_{AB}(P_{\a|x}^A \otimes P_{\b|y}^B))$, where $\rho_{AB}$ is the measured quantum state and $P_{\a|x}^A, P_{\b|y}^B$ are POVM elements on subsystems $A$ and $B$, corresponding to the measurement choices $x$ and $y$ and outcomes $\a$ and $\b$. Indeed, any binary outcome POVM $P_{\s|t}$ verifying $P_{\s|t}\succeq 0$ and $P_{0|t}+P_{1|t}=\id$ is in one-to-one correspondence with a binary measurement operator $X_t = P_{0|t}-P_{1|t}$ verifying $-\id \preceq X_t \preceq \id$. Therefore, it is possible to express the statistics in terms of the components of the correlation vector $\Vec P$ as
\begin{equation}
p(\a\b|xy)=\frac{1 + \a\cor{A_x} + \b\cor{B_y} + \a\b\cor{A_x B_y}}{4}.
\end{equation}
Note that the measurement operator $X_t$ is unitary iff the corresponding POVM is a projector.
\end{Remark}

\begin{Lemma}[Naimark's dilation theorem for binary measurement operators] \label{lemma:Naimark}
For any binary measurement operator $A$ acting on a Hilbert space $\mathcal{H}$, there exists a Hilbert space $\Tilde{\mathcal{H}}$, an isometry $V:\mathcal{H}\to \Tilde{\mathcal{H}}$ and a unitary binary measurement operator $\Tilde{A}$ s.t.
\begin{equation} \label{eq:dilation}
    A = V^\dagger \Tilde{A} V. 
\end{equation}
\end{Lemma}
\begin{proof}
    Consider the tensor product Hilbert space $\Tilde{\mathcal{H}} = \mathcal{H}\otimes \mathbb{C}^2$ and introduce the following isometry and unitary binary measurement operator:
    \begin{eqnarray}
        V = \sum_{i=0}^1 \sqrt{\frac{1+(-1)^i A}{2}} \otimes \ket{i}_{\mathbb{C}^2}, \quad
        \Tilde{A} = \id_\mathcal{H} \otimes \sigma_z. 
    \end{eqnarray}
    On can easily verify that the first term is well-defined as both $\frac{1 \pm A}{2}$ are positive operators, and $V^\dagger V = \id_\mathcal{H}$ so $V$ is an isometry, that $\Tilde{A}^2 = \id_{\Tilde{\mathcal{H}}}$ and that property \cref{eq:dilation} is satisfied.
\end{proof}

\begin{Remark}
It is well known that Naimark's dilation theorem also applies to measurements involving an arbitrary number of measurements outcomes $k\geq 1$, see e.g.~\cite[Proposition 9.6]{Harris16}.
\end{Remark}

\begin{Corollary}[Naimark's dilation for two measurements] \label{cor:Naimarkfortwo}
    For any two binary measurement described by operators $A_0$, $A_1$ acting on a Hilbert space $\mathcal{H}$, there exists a Hilbert space $\Tilde{\mathcal{H}}$, an isometry $V:\mathcal{H}\to \Tilde{\mathcal{H}}$ and unitary binary measurement operators $\Tilde{A}_x$ s.t.
    \begin{equation}
        \forall x\in\{0,1\}, \ A_x = V^\dagger \Tilde{A}_x V. 
    \end{equation}
\end{Corollary}

\begin{proof}
    We can apply Naimark's dilation for the operator $A_0$. Therefore, there exists a Hilbert space $\mathcal{H}_1$, an isometry $V_1:\mathcal{H}\to \mathcal{H}_1$ and a unitary operator $A_0'$ s.t.~$A_0 = V_1^\dagger A_0' V_1$. Then one can define a measurement operator $A_1'= V_1 A_1 V_1^\dagger$ on $\mathcal{H}_1$, verifying $-\id \preceq A_1' \preceq \id$. One can thus apply Naimark's dilation to $A_1'$. This guarantees the existence of a Hilbert space $\mathcal{H}_2$, an isometry $V_2:\mathcal{H}_1\to \mathcal{H}_2$ and a unitary operator $\Tilde{A}_1$ s.t.~$A_1' = V_2^\dagger \Tilde{A}_1 V_2$. Let's define $\Tilde A_0 = V_2 A_0' V_2^\dagger + V_2 V_2^\dagger - \id_{\mathcal{H}_2}$. Since $V_2^\dagger V_2 = \id_{\mathcal{H}_1}$, we have $\Tilde{A}_0^2 = \id_{\mathcal{H}_2}$, so $\Tilde{A}_0$ is a unitary operator. Now, we can define $V:=V_2\circ V_1$, verifying
    \begin{equation}
         \forall x\in\{0,1\}, \ A_x = V^\dagger \Tilde{A}_x V
    \end{equation}
    where $\Tilde{A}_x$ describe projective measurements.
\end{proof}

\begin{Remark}
Naimark's dilation can be further extended to an arbitrary family of $m\geq 1$ measurements by repeating the above construction iteratively, see e.g.~\cite[Theorem 9.8]{Harris16}.
\end{Remark}

\begin{Corollary}[Naimark's dilation for the CHSH scenario] \label{cor:corr_conservation}
    Any quantum correlation in the CHSH scenario can be realized by \emph{projective} measurements, i.e.~for all $\Vec P \in \mathcal{Q}$, there exists a quantum state $\rho_{AB}$ and unitary binary measurement operators $(\Tilde{A}_x, \Tilde{B}_y)$ such that $\Vec P = \{\Tr(\rho_{AB} (\Tilde{A}_x \otimes \Tilde{B}_y))\}_{xy}$.
\end{Corollary}

\begin{proof}
    Consider a vector $\Vec P \in \mathcal{Q}$. It admits a realization involving a shared state $\sigma_{AB}$ and binary measurement operators $(A_x, B_y)$. Using \cref{cor:Naimarkfortwo}, one can define isometries $V_{A(B)} : \mathcal{H}_{A(B)} \to \Tilde{\mathcal{H}}_{A(B)}$ such that for all $x,y$, $A_x = V_A^\dagger \Tilde{A}_x V_A$ and $B_y = V_B^\dagger \Tilde{B}_y V_B$ with $\Tilde{A}_x$ and $\Tilde{B}_y$ unitary binary measurement operators. Then the coefficients of $\Vec P$ are given by
    \begin{equation}
    \begin{split}
        \{\Vec P\}_{xy} & = \Tr(\sigma_{AB}(A_x\otimes B_y)) \\
        &= \Tr(\sigma_{AB}(V_A^\dagger \Tilde{A}_x V_A \otimes V_B^\dagger \Tilde{B}_y V_B)) \\
        & = \Tr((V_A \otimes V_B)\sigma_{AB}(V_A^\dagger \otimes V_B^\dagger)(\Tilde{A}_x \otimes \Tilde{B}_y))\\
        & = \Tr(\rho_{AB} (\Tilde{A}_x \otimes \Tilde{B}_y)),
    \end{split}
    \end{equation}
    where we defined the quantum state $\rho_{AB}:=(V_A \otimes V_B)\sigma_{AB}(V_A \otimes V_B)^\dagger$ on the Hilbert space $\Tilde{\mathcal{H}}_{A}\otimes \Tilde{\mathcal{H}}_{B}$.
\end{proof}

\begin{Remark}
Due to the previous extension, \cref{cor:corr_conservation} can be expanded to an arbitrary Bell scenario, in which $n$ parties each have a $m_i$ measurement settings and each measurement has $k_{i,j_i}$ possible outcomes, where $i=1,\ldots,n$ and $j_i=1,\ldots,m_i$.
\end{Remark}

\begin{Remark}
At this stage, one could also apply a dilation to the state $\rho_{AB}$ to show that any point $\Vec P \in \mathcal{Q}$ can be realized by measuring a \emph{pure} state with projective measurements.
\end{Remark}

\begin{Lemma}[Jordan's lemma] \label{lemma:Jordan}
For any pair of projective binary measurements $A_0, A_1$ acting on a Hilbert space $\mathcal{H}$, there exists an orthogonal basis of $\mathcal{H}$ such that $A_0$ and $A_1$ are block-diagonal with blocks of size at most two. Furthermore, in this basis, $A_0$ and $A_1$ are real operators.
\end{Lemma}

\begin{proof}
    Consider the hermitian operator $R = A_0 + A_1$. For any eigenvector $\ket{\psi}$ of $R$ of eigenvalue $\lambda \in \mathbb{R}$, one can consider the subspace $V_{\psi} = \text{Vect}\langle \ket{\psi}, A_0 \ket{\psi} \rangle$. Then, one of the two following cases happens: 
    \begin{enumerate}
        \item[Case 1 :] $\dim V_\psi = 1$. Then there exist $\mu \in \mathbb{C}$ s.t.~$A_0 \ket{\psi} = \mu \ket{\psi}$. As $A_0$ is hermitian, $\mu$ has to be real. Then, $A_1 \ket{\psi} = R\ket{\psi} - A_0 \ket{\psi}=(\lambda-\mu) \ket{\psi}$. So $\ket{\psi}$ is an eigenvector of both $A_0$ and $A_1$ of real eigenvalue (block of size $1$). \\
        \item[Case 2 :] $\dim V_\psi = 2$. Then, since $A_0^2=\id$, $V_\psi$ is stabilized by $A_0$ and in this subspace $A_0|_{V_\psi} = \sigma_x$. Furthermore, $A_1 \ket{\psi} = \lambda \ket{\psi} - A_0 \ket{\psi}$ belongs to $V_\psi$ and since $A_1^2 = \id$, we have that $A_1(A_0\ket{\psi}) = \lambda A_1 \ket{\psi} - \ket{\psi}$ also belongs to $V_\psi$. Finally, $V_\psi$ is stabilized by both $A_0$ and $A_1$ (block of size 2) and all have real coefficients in the sub-basis.
    \end{enumerate}
    The rest of the proof goes on finite recursion on the eigenvectors of $R$. At each step, one has to find an eigenvector $\ket{\psi'}$ orthogonal to all previous subspaces, which is always possible since $R$ is hermitian.
\end{proof}

\begin{Corollary}[Jordan's lemma for the CHSH scenario]\label{cor:JordanBell}
    Every quantum correlations in the CHSH scenario can be obtained as a convex combination of realizations involving local Hilbert spaces of dimension $2$ and real projective measurements.
\end{Corollary}

\begin{proof}
    For $\Vec P\in \mathcal{Q}$, we consider a projective realization thanks to \cref{cor:corr_conservation}. Furthermore, one can assume, up to trivial dilation, that the dimensions of $\mathcal{H}_{A(B)}$ are even. We can then apply Jordan's lemma for each party and assume that all block are of size 2. Thus, there exists local projections $\Pi_A^i$, $\Pi_B^j$ on $\mathcal{H}_{A(B)}$ respectively such that $A_x^{(i)} := \Pi_A^i A_x \Pi_A^i$ are (2,2)-real matrices and $A_x = \bigoplus_i A_x^{(i)}$, likewise for Bob. As such
    \begin{equation}
    \begin{split}
        \{\Vec P\}_{xy} & = \Tr(\rho_{AB}(A_x \otimes B_y)) \\
        & = \Tr(\rho_{AB}(\bigoplus_i A_x^{(i)} \otimes \bigoplus_j B_y^j)) \\
        & = \sum_{ij} \Tr((\Pi_A^i \otimes \Pi_B^j) \rho_{AB} (\Pi_A^i \otimes \Pi_B^j)(A_x^{(i)} \otimes B_y^j)) \\
        & = \sum_{ij} p^{ij} \{\Vec P_{ij}\}_{xy},
    \end{split}
    \end{equation}
    where we defined $p^{ij}:=\Tr((\Pi_A^i \otimes \Pi_B^j) \rho_{AB} (\Pi_A^i \otimes \Pi_B^j))$, $\{\Vec P_{ij}\}_{xy}=\Tr(\rho_{AB}^{ij}(A_x^{(i)} \otimes B_y^j))$ and $\rho_{AB}^{ij}:=(\Pi_A^i \otimes \Pi_B^j) \rho_{AB} (\Pi_A^i \otimes \Pi_B^j)/p^{ij}$, which verify $\Tr(\rho_{AB}^{ij})=1$, $\rho_{AB}^{ij} \succeq 0$, $p^{ij} \geq 0$ and $\sum_{ij} p^{ij}=1$. 
\end{proof}

\begin{Remark}
\cref{cor:JordanBell} generalizes straightforwardly to Bell scenarios involving $n$ parties with two binary measurements per party.
\end{Remark}

\begin{Corollary}[Jordan's lemma for extremal correlations in the CHSH scenario]\label{cor:JordanBellExtr}
    Every extremal quantum correlations in the CHSH scenario can be realized by measuring a single two-qubit state with real projective measurements.
\end{Corollary}

\begin{proof}
    For an extremal point $\Vec P \in \text{Ext}(\mathcal{Q})$, we can apply the previous result to decompose $\Vec P$ as a convex combination of correlations $\Vec P_{ij}$ obtained by measuring two qubit states with real unitary operators. The extremality of $\Vec P$ ensures that for all $i,j$, $\Vec P_{ij} = \Vec P$ and as such any qubit realization of a given $\Vec P_{ij}$ provides a satisfying realization for $\Vec P$. 
\end{proof}

\begin{Proposition}[Extremal points realization in the CHSH scenario] \label{prop:realqubitreduction}
    Every extremal quantum correlations in the CHSH scenario can be realized by measuring a real, pure, two-qubit state with real projective measurements.
\end{Proposition}

\begin{proof}
    For $\Vec P \in \text{Ext}(\mathcal{Q})$, one can apply \cref{cor:JordanBellExtr} to obtain a realization of the form $(\rho_{AB}, A_x, B_y)$ with $A_x, B_y$ real unitary operators and local Hilbert spaces of dimension 2. If $\rho_{AB}$ is not equal to its complex conjugate $\Bar{\rho}_{AB}$, then one can also realize $\Vec P$ by measuring the real state $\Tilde \rho_{AB} = (\rho_{AB}+\Bar{\rho}_{AB})/2$, because all measurement operators are real. Therefore, we can assume that $\rho_{AB}$ is real. A real density matrix $\rho_{AB}$ can always be written as
    \begin{equation}
        \rho_{AB} = \sum_i \lambda_i \ketbra{\psi_i}
    \end{equation}
    with $\lambda_i \in \mathbb{R}_+$, $\sum_i \lambda_i = 1$, and where $\ket{\psi_i}$ are real two-qubit states. Therefore, one has $\Vec P = \sum_i \lambda_i \Vec P_i$, where $\Vec P_i$ is obtained by measuring $\ketbra{\psi_i}$ on $(A_x, B_y)$. The extremality of $\Vec P$ ensures that for all $i$, $\Vec P_i = \Vec P$, and therefore the realization for any $i$ is a valid realization for $\Vec P$.
\end{proof}

\begin{Remark}
An alternative proof of \cref{prop:realqubitreduction} was recently given in \cite[Appendix A]{Mikosnuszkiewicz23}.
\end{Remark}

\begin{Remark}
For the study of extremal point only, one does not necessarily need Naimark's dilation to impose that the measurement operators are unitary (see \cite[Lemma 1]{Masanes05} for an alternative argument). 
\end{Remark}

\begin{Remark}
Corollaries \ref{cor:JordanBell}, \ref{cor:JordanBellExtr} and \cref{prop:realqubitreduction} can be directly extended to Bell scenarios in which $n$ parties each have two binary outcome measurements, proving that in those scenarios all extremal quantum correlations can be achieved by measuring real $n$-qubit states with real unitary measurements.
\end{Remark}

\begin{Lemma} [Extremality under relabelings] \label{lemma:symmetries}
    Consider a convex set $\mathcal{K}\subset \mathbb{R}^n$ stabilized by a linear involution $S:\mathbb{R}^n \to \mathbb{R}^n$, i.e.~such that $S^2 = 1$ and $S(\mathcal{K}) = \mathcal{K}$. Then, any point $\Vec P \in \mathcal{K}$ is extremal iff $S(\Vec P)$ is extremal. 
\end{Lemma}

\begin{proof}
    Suppose that $\Vec P$ is extremal. For any convex decomposition $S(\Vec P) = \lambda \Vec P_1  + (1-\lambda) \Vec P_2$ of $S(\Vec P)$, we have $\Vec P = \lambda S(\Vec P_1)  + (1-\lambda) S(\Vec P_2)$, because $S$ is linear and is an involution. By extremality of $\Vec P$, we have $\Vec P = S(\Vec P_1) = S(\Vec P_2)$, and as such $S(\Vec P) = \Vec P_1 = \Vec P_2$ so $S(\Vec P)$ is extremal. Conversely, one can use that $\Vec P = S(S(\Vec P))$. 
\end{proof}

\begin{Remark}
    The quantum set is invariant under the following involutions:
    \begin{subequations}
        \begin{align}
            \text{relabeling of parties: }& S_{part} : A_x \longleftrightarrow B_x, \\
            \text{relabeling of inputs: }& S_{in} : A_0 \longleftrightarrow A_1, \\
            \text{relabeling of outputs: }& S_{out} : A_0 \longrightarrow -A_0.
        \end{align}
    \end{subequations}
    As such, the study of whether a quantum correlation $\Vec P$ is extremal can be reduced to the study of the extremality of any correlations of the form $g\cdot \Vec P$ where $g$ is an arbitrary element of the Bell group $G$ generated by $S_{part},S_{in}, S_{out}$ (a group of order 32). 
\end{Remark}

\begin{Proposition} [Sufficient parametrization of extremal points in the CHSH scenario] \label{prop:qubitrealization}
    The study of the extremal points of $\mathcal{Q}$ can be reduced to the study of the extremality of correlations admitting a realization of the form:
    \begin{subequations} \label{eq:qubitrealization}
        \begin{align}
            \text{state: }& \ket{\phi_\theta} = \cos(\theta)\ket{00} + \sin(\theta)\ket{11}, \\
            \text{measurements: }& A_x = \cos(a_x) \sigma_z + \sin(a_x) \sigma_x, \\
            & B_y = \cos(b_y) \sigma_z + \sin(b_y) \sigma_x.
        \end{align}
    \end{subequations}
    where the real parameters $\theta, a_x, b_y$ can be reduced, upon relabeling symmetries, to the range
    \begin{equation}
        \theta, a_x, b_y \in [0,\pi), \ a_0 \leq a_1, \ b_0 \leq b_1, \ a_0 \leq b_0. 
    \end{equation}
\end{Proposition}

\begin{proof}
    Thanks to \cref{prop:realqubitreduction}, all extremal points of $\mathcal{Q}$ admit a realization of the form \cref{eq:qubitrealization} for arbitrary real parameters $\theta, a_x, b_y$. The periodicity of the corresponding correlations allows reducing those to the interval $\theta \in [0,\pi)$ and $a_x, b_y \in [0,2\pi)$. Using \cref{lemma:symmetries}, one can study the nature of such correlations by applying arbitrary symmetries of $\mathcal{Q}$. In particular, one can use relabeling of outcomes to ensure that $a_x, b_y \in [0,\pi)$, relabeling of inputs to have $a_0 \leq a_1$ and $b_0 \leq b_1$ and relabeling of parties for $a_0 \leq b_0$. 
\end{proof}

\section{Steering of quantum realizations} \label{sec:steering}

In this appendix, we describe in more details the steering map introduced in \cref{fig:nonlinearMap} of the main text. In the first section, we consider its action on vectors of $\mathbb{C}^2$. In the second section, we discuss the action of the map on qubit measurements and relate the statistics before and after the transformation. In both cases, our analysis applies to arbitrary qubit states and measurement settings (i.e.~not necessarily real ones).

\subsection{Steering of states} \label{sec:statesteering}

The non-linear steering map $T_{\theta}:\mathbb{C}^2\to \mathbb{C}^2$ is defined by the action
\begin{equation} \label{eq:steeringtrans}
T_{\theta}:\ket{u}\mapsto \frac{\sqrt{2} (\braket{\phi_\theta}{u})\ket{\phi^+}}{||\braket{\phi_\theta}{u}||}
\end{equation}
on state $\ket{u} \in \mathcal{H}_A=\mathbb{C}^2$. In particular, in the canonical parametrization
\begin{equation}\label{eq:uVector}
\ket{u}=\cos(\frac{a}2) \ket{0} + e^{i\psi} \sin(\frac{a}2) \ket{1},
\end{equation}
the steering transformation returns
\begin{equation} \label{eq:explicitsteering}
    T_\theta(\ket{u}) = \frac{\cos(\theta) \cos(\frac{a}2) \ket{0} + e^{i\psi} \sin(\theta) \sin(\frac{a}2) \ket{1}}{\sqrt{\cos(\theta)^2 \cos(\frac{a}2)^2 + \sin(\theta)^2 \sin(\frac{a}2)^2}} = \epsilon_\theta\left(\cos(\frac{\Tilde{a}}2) \ket{0} + e^{i\psi}\sin(\frac{\Tilde{a}}2) \ket{1}\right)
\end{equation}
where $\epsilon_\theta := \text{sgn}(\cos(\theta))$ and $\Tilde{\alpha}$ is the transformed state angle. Note that the denominator can only cancel if $\theta \in \{0, \pi/2\}$ and $a=2\theta \pm\pi$. Whenever this happens, we define $T_\theta(\ket{u})$ to be the null vector (since this only happens when $\ket{\phi_\theta}$ is a product state, we won't be interested in this case). The angle of the image state in the canonical basis is given by 
\begin{equation}
    \Tilde{\alpha}:= 2 \operatorname{atan}(\tan(\frac{a}2)\tan(\theta)) \in [-\pi, \pi].
\end{equation}

Note that in all cases the image vectors satisfy $\langle \phi_\theta | u \rangle = \sqrt{2} ||\braket{\phi_\theta}{u}|| \langle \phi^+ | T_\theta(\ket{u}) \rangle $, and as such
\begin{equation} \label{eq:correlationtransform}
    \begin{split}
        \forall B \in L(\mathcal{H}_B), \ \bra{\phi_\theta} \left(\ketbra{u} \otimes B \right)\ket{\phi_\theta} =  {2 ||\braket{\phi_\theta}{u}||^2} \bra{\phi^+} \left(\ketbra{v} \otimes B \right) \ket{\phi^+},
    \end{split}
\end{equation}
where for clarity we denoted $\ket{v}=T_\theta(\ket{u})$. Finally, one can notice that the inverse of the steering function is easy to compute. Indeed, we have:
\begin{equation}
    T_{\pi/2 - \theta} (T_\theta(\ket{u})) = \ket{u}.
\end{equation}

\subsection{Steering of correlations and realizations} \label{sec:realizationsteering}
When considering binary non-degenerate projective measurements acting on a pure state with local Hilbert spaces of dimension 2, the measured state can always be written as $\ket{\phi_\theta}$ in the Schmidt basis. In this basis, the projective measurements can be written as
\begin{equation}
    \begin{split}
        & A_x = \cos(a_x) \sigma_z + \sin(a_x) (\cos(\alpha_x) \sigma_x + \sin(\alpha_x) \sigma_y), \\
        & B_y = \cos(b_y) \sigma_z + \sin(b_y) (\cos(\beta_y) \sigma_x + \sin(\beta_y) \sigma_y),
    \end{split}
\end{equation}
and can be decomposed as the difference of two orthogonal projectors:
\begin{equation}
    \begin{split}
        & A_x = \ketbra{u_{+,x}} - \ketbra{u_{-,x}}, \\
        & B_y = \ketbra{w_{+,y}} - \ketbra{w_{-,y}},
    \end{split}
\end{equation}
where $\braket{u_{\a,x}}{u_{\a',x}}=\delta_{\a \a'}$, $\braket{w_{\b,x}}{w_{\b',x}}=\delta_{\b \b'}$ and $\ket{u_{\a,x}}, \ket{w_{\b,y}} \in \mathbb{C}^2$ are in general complex vectors of the form \cref{eq:uVector} with $a\in(-\pi,\pi]$. 

The expected probabilities $p(\a\b|xy) = \frac{1}{4}(1+\a \cor{A_x}+\b\cor{B_y}+\a \b \cor{A_x B_y})$ corresponding to this realization are given by
\begin{equation}
    p(\a\b|xy) = \bra{\phi_\theta} \left(\ketbra{u_{\a,x}} \otimes \ketbra{w_{\b,y}} \right) \ket{\phi_\theta}.
\end{equation}
Applying the steering map on the eigenstates $\ket{u_{\a,x}}$ of Alice's measurements and using \cref{eq:correlationtransform}, we obtain:
\begin{equation}
    p(\a\b|xy) = {2p(\a|x)}\bra{\phi^+} \left(\ketbra{v_{\a,x}} \otimes \ketbra{w_{\b,y}} \right) \ket{\phi^+},
\end{equation}
where $\ket{v_{\a,x}}=T_\theta(\ket{u_{\a,x}})$ for all $\a,x$.

Note that the vectors $\ket{v_{\a,x}}$ are units or null but in full generality $\ket{v_{1,x}}$ and $\ket{v_{-1,x}}$ are not orthogonal anymore. We thus define a new unitary operator on Alice's Hilbert space for each vector $\ket{v_{\a,x}}$:
\begin{equation} \label{eq:newoperators}
    \Tilde{A}^{\a}_x = \a \left(2\ket{v_{\a,x}}\bra{v_{\a,x}} - \id\right).
\end{equation}
These new operators verify that for all $\a,x,y$: 
\begin{equation} \label{eq:newoperatorsdef}
    \frac{\cor{ \phi_\theta | A_x B_y | \phi_\theta} + \a \cor{ \phi_\theta | B_y | \phi_\theta}}{ 1+\a \cor{ \phi_\theta | A_x | \phi_\theta}} = \bra{\phi^+} \Tilde{A}^{\a}_x B_y \ket{\phi^+} =: [\Tilde{A}^{\a}_x B_y], 
\end{equation}
which is well-defined as long as $\cor{A_x}\neq \pm1$, i.e.~at least whenever $\theta\notin \{0,\pi/2\}$.

Upon non-linear transformation $T_\theta$, the original correlation vector $\Vec P$, obtained with a pure and projective two qubit realization, can thus be interpreted as correlations of a quantum point $\Vec Q$ realized by the state $\ket{\phi^+}$, four measurements $\Tilde{A}^{\a}_x$ for Alice and two measurements $B_y$ for Bob:
\begin{equation}
\Vec P = \begin{array}{c|c|c}
          & \langle B_0\rangle  & \langle B_1\rangle \\
         \hline
         \langle A_0\rangle  & \langle A_0 B_0\rangle & \langle A_0 B_1\rangle \\
         \hline
         \langle A_1\rangle & \langle A_1 B_0\rangle & \langle A_1 B_1\rangle
    \end{array} \xrightarrow{\quad T_\theta\quad } \Vec Q = \begin{array}{c|c|c}
          & 0 & 0 \\
         \hline
         0 & [\Tilde{A}^{+}_{0} B_0] & [\Tilde{A}^{+}_{0} B_1] \\
         \hline 
         0 & [\Tilde{A}^{-}_{0} B_0] & [\Tilde{A}^{-}_{0} B_1] \\
         \hline 
         0 & [\Tilde{A}^{+}_{1} B_0] & [\Tilde{A}^{+}_{1} B_1] \\
         \hline 
         0 & [\Tilde{A}^{-}_{1} B_0] & [\Tilde{A}^{-}_{1} B_1] 
    \end{array}
\end{equation}
Note that all marginal terms of this new quantum point are obtained using the property that marginals correlations are $0$ when measuring $\ket{\phi^+}$: $\bra{\phi^+} \Tilde{A}^{\a}_x \ket{\phi^+} = \bra{\phi^+} B_y \ket{\phi^+} = 0$, therefore both behaviors have the same number of degrees of freedom.

The point $\Vec Q$ admits a realization that can be parameterized by
\begin{equation} \label{eq:realizationQ}
        \begin{split}
            \text{state: } & \ket{\phi^+} = \frac{1}{\sqrt{2}}(\ket{00}+\ket{11}), \\
            \text{measurements: } & \Tilde{A}^{\a}_x = \cos(\Tilde{a}^{\a}_x) \sigma_z + \sin(\Tilde{a}^{\a}_x) (\cos(\alpha_x) \sigma_x + \sin(\alpha_x) \sigma_y), \ \Tilde{a}^{\a}_x \in [0,2\pi), \\
            & B_y = \cos(b_y) \sigma_z + \sin(b_y) (\cos(\beta_y) \sigma_x + \sin(\beta_y) \sigma_y), \ b_y \in [0,2\pi),\\
        \end{split}
\end{equation}
where if one assumes that $a_x \in [0,\pi)$ the new measurement angles are given by
\begin{equation} \label{eq:anglemodif}
    \Tilde{a}^{\a}_x = 2  \operatorname{atan}\left(\tan(\frac{a_x}{2})\tan(\theta)^\a \right).
\end{equation}

Note that when the initial state is the singlet state ($\theta=\pi/4$), then the modified measurements are the same, i.e.~$\Tilde{a}^{\a}_x = a_x$ for all $\a, x$. The limit $\theta\rightarrow 0$ is well-defined and gives modified angles $\Tilde a_x^{-1}=0$ when $a_x=0$ and $\Tilde a_x^{-1}=\pi$ otherwise. Furthermore, when $\theta \in [0, \pi/4]$, we have
\begin{subequations}
    \begin{align}
        & 0 \leq \Tilde{a}_x^+ \leq a_x \leq \Tilde{a}_x^- \leq \pi, \\
        & a_x \longrightarrow \Tilde a_x^+ \text{ is a increasing function of $\theta$}, \\
        & a_x \longrightarrow \Tilde a_x^- \text{ is an decreasing function of $\theta$}.
    \end{align}
\end{subequations}
and when $\theta \in (\pi/4, \pi/2)$: 
\begin{subequations}
    \begin{align}
        & 0 \leq \Tilde{a}_x^- \leq a_x \leq \Tilde{a}_x^+ \leq \pi, \\
        & a_x \longrightarrow \Tilde a_x^+ \text{ is a decreasing function of $\theta$}, \\
        & a_x \longrightarrow \Tilde a_x^- \text{ is an increasing function of $\theta$}.
    \end{align}
\end{subequations}

\subsection{Steering on Alice vs Bob} \label{sec:AliceVSBob}

Everything in this appendix about the steering transformation could be done on Bob's side as well. As such, it is possible to associate to each measurement $B_y$ two modified measurements $\tilde B_y^\b$ on Bob's Hilbert space such that one can map the initial realization on a partially entangled state to a realization on a maximally entangled state. Those modified measurements would also have modified angles with respect to $\sigma_z$, following the same transformation
\begin{equation} 
    \Tilde{b}^{\b}_y = 2  \operatorname{atan}\left(\tan(\frac{b_y}{2})\tan(\theta)^\b \right),
\end{equation}
for all $\b,y$. The realization associated to these measurements $(A_x, \tilde B_y^\b)$ with angles $a_x$ and $\tilde b_y^\b$ is in general different from the one involving $(\tilde A_x^\a, B_y)$ where the steering is applied on Alice's side: the value of the correlators $[A_x \tilde B_y^\b]$ take different values than the $[\tilde A_x^\a B_y]$.

However, the fact that the modified measurements angles are fully alternating is unchanged by whether the transformation is applied on Alice or Bob's side. Indeed, for a fixed entanglement parameter $\theta$, the steering transformation of the measurements angles verify the following:
\begin{enumerate}
    \item It is monotonous, i.e. for all $\s \in \{ \pm1 \}$, $a \leq b \Longrightarrow [\tilde a^\s]_\pi \leq [\tilde b^\s]_\pi$.
    \item Modifying an angle twice in a row with opposite outcome sign gives back the original angle, as $2\operatorname{atan}(\tan([2\operatorname{atan}(\tan(a/2)\tan(\theta))]/2)\tan(\theta)^{-1})=a$.
\end{enumerate}
With those two properties, one can verify that
\begin{equation}
    \forall a,b \in [0,\pi),\ [a^\s]_\pi \leq b \Longleftrightarrow a \leq [b^{-\s}]_\pi
\end{equation}
by application of the steering transformation. As such, the fully alternating condition when the steering is applied on Alice's measurement angles is equivalent to the full alternating condition after application of the steering condition on Bob's side:
\begin{equation}
    \left( \forall \s,\t, \  0 \leq [a_0^\s]_\pi \leq b_0 \leq [a_1^\t]_\pi \leq b_1 < \pi \right) \Longleftrightarrow \left( \forall \s,\t, \  0 \leq a_0 \leq [b_0^\s]_\pi \leq a_1 \leq [b_1^\t]_\pi < \pi \right).
\end{equation}

\subsection{Proof of Proposition 1} \label{sec:proofprop1}

A subset of $\mathcal{Q}$ that is very well characterized is the subspace of zero marginals distributions. In this case, the following statement holds:
\begin{Proposition}[Masanes Theorem, \cite{Masanes03}] \label{prop:zeromargarcsin}
    In the CHSH scenario, the set of quantum points with zero marginals is exactly the set of points verifying the following eight inequalities:
    \begin{equation} \label{eq:ineqcorrelation}
\begin{array}{r}
-\pi \leq-\operatorname{asin} \langle A_0B_0 \rangle+\operatorname{asin} \langle A_1B_0\rangle+\operatorname{asin} \langle A_0B_1\rangle+\operatorname{asin}\langle A_1B_1 \rangle\leq \pi \\
-\pi \leq \operatorname{asin}\langle A_0B_0\rangle-\operatorname{asin}\langle A_1B_0\rangle+\operatorname{asin} \langle A_0B_1\rangle+\operatorname{asin} \langle A_1B_1 \rangle \leq \pi \\
-\pi \leq \operatorname{asin}\langle A_0B_0\rangle+\operatorname{asin} \langle A_1B_0\rangle-\operatorname{asin} \langle A_0B_1\rangle+\operatorname{asin} \langle A_1B_1\rangle\leq \pi \\
-\pi \leq \operatorname{asin}\langle A_0B_0\rangle+\operatorname{asin} \langle A_1B_0\rangle+\operatorname{asin}\langle A_0B_1\rangle-\operatorname{asin} \langle A_1B_1\rangle\leq \pi.
\end{array}
\end{equation}
\end{Proposition}
Since any point obtained by measuring a maximally entangled two qubits state $\ket{\phi^+}= (\ket{00}+\ket{11})/\sqrt{2}$ has zero marginals, it must verify the above inequalities.

\begin{proof}[Proof of Proposition 1]
    For any point $\Vec P \in \mathcal{Q}$ which admits a pure projective two qubit realization, the steering transformation allows one to consider a non-linear transformation of $\Vec P$ granting a vector of correlation $\Vec Q = \{[A_x^\a B_y]\}_{\a,x,y}$ with zero marginals, in a scenario with four binary measurements for Alice and two binary measurement for Bob. Thus, for each pair $(A_0^\s, A_1^\t)$ of Alice's new measurements, \cref{prop:zeromargarcsin} holds and  grants:
    \begin{equation} \label{eq:ineqmargin}
\begin{array}{r}
-\pi \leq-\operatorname{asin} [A_0^\s B_0] +\operatorname{asin} [A_1^\t B_0] +\operatorname{asin} [A_0^\s B_1]+\operatorname{asin}[A_1^\t B_1]\leq \pi \\
-\pi \leq \operatorname{asin} [A_0^\s B_0] -\operatorname{asin} [A_1^\t B_0] +\operatorname{asin} [A_0^\s B_1]+\operatorname{asin}[A_1^\t B_1] \leq \pi \\
-\pi \leq \operatorname{asin} [A_0^\s B_0] +\operatorname{asin} [A_1^\t B_0] -\operatorname{asin} [A_0^\s B_1]+\operatorname{asin}[A_1^\t B_1]\leq \pi \\
-\pi \leq \operatorname{asin} [A_0^\s B_0] +\operatorname{asin} [A_1^\t B_0] +\operatorname{asin} [A_0^\s B_1]-\operatorname{asin}[A_1^\t B_1]\leq \pi
\end{array}
\end{equation}
for all $\s,\t \in \{-1,1\}$, where: 
\begin{equation} \label{eq:newcorrphi+}
    [A_x^\a B_y] = \frac{\langle A_x B_y \rangle+\a \langle B_y \rangle}{ 1+\a \langle A_x \rangle}
\end{equation}
\end{proof}

\section{Proof of \cref{thm:self-test}} \label{sec:proofLemma1}

This section is organized in the following way: in the first two subsections we prove some preliminary results. In particular, we prove in \cref{lemma:uniqueness} that the conditions of \cref{thm:self-test} along with the assumption that the underlying realization is pure, projective and and nondegenerate in $\mathcal{Q}_2$ is sufficient to fully identify a unique realization. The third subsection is dedicated to the proof of \cref{thm:self-test}. To prove extremality, we consider convex decompositions into sub-correlations. We make use of the steering transformation, introduced in~\cref{sec:steering}, to identify four underlying realizations on four different states. We then use geometrical considerations to understand how the conditions \cref{eq:conditions} relate these different realizations to each other. Finally, we transpose these conclusions onto the original quantum realization and infer its properties.

\subsection{Preliminary Geometrical Considerations}
Let us recall some elements of geometry. Here, we consider normalized states $\ket{\psi_1},\ket{\psi_2},\ket{\psi_3}$ in a real Hilbert space $\mathcal{H}=\mathbb{R}^d$.
\begin{Definition}
The angular distance between $\ket{\psi_1}$ and $\ket{\psi_2}$ is defined as the angle $\theta\in[0,\pi]$ s.t.
\begin{equation}
\cos\theta=\braket{\psi_1}{\psi_2}.
\end{equation}
\end{Definition}
\noindent The angular distance $\theta$ is 0 iff $\ket{\psi_1}=\ket{\psi_2}$, $\pi$ iff $\ket{\psi_1}=-\ket{\psi_2}$, and $\pi/2$ iff $\ket{\phi_1}$ and $\ket{\psi_2}$ are orthogonal.

\begin{Lemma} \label{lemma:triangle}
The pairwise angular distance between three vectors $\ket{\psi_1}, \ket{\psi_2}, \ket{\psi_3}$ satisfies the triangle inequality
\begin{equation}
|\theta_{12}-\theta_{23}| \leq \theta_{13} \leq \theta_{12}+\theta_{23}, \label{eq:triangle}
\end{equation}
with equality only if the vectors are coplanar.
\end{Lemma}
\begin{proof}
Clearly, two vectors define a plane, and therefore if either pair of states are colinear ($\theta=0,\pi$), then all vectors are coplanar. Furthermore, whenever $\theta_{ij}=0$ for some $i\neq j$, then $\theta_{ik}=\theta_{jk}$ ($k\neq i,j$) and all versions of the inequalities \cref{eq:triangle} are verified. Similarly, whenever $\theta_{ij}=\pi$ for some $i\neq j$, then $\theta_{ik}=\pi-\theta_{jk}$ ($k\neq i,j$) and all versions of the inequalities \cref{eq:triangle} are verified. We can thus assume that the angular distance between every pair of states is strictly contained in the interval $0<\theta_{ij}<\pi$.

In this case, $\ket{\psi_1}$ and $\ket{\psi_2}$ define a plane. In particular, there exist two orthonormal vectors $\ket{\phi_1}$, $\ket{\phi_2}$ such that
\begin{align}
\ket{\psi_2}&=\ket{\phi_1}\\
\ket{\psi_1}&=\cos\theta_{12}\ket{\phi_1} + \sin\theta_{12}\ket{\phi_2}.
\end{align}
Without loss of generality, the third vector can then be written
\begin{equation}
\ket{\psi_3} = \cos\theta_{23}\ket{\phi_1} + \sin\theta_{23}\cos\alpha\ket{\phi_2} + \sin\theta_{23}\sin\alpha\ket{\phi_3}
\end{equation}
for $\ket{\phi_3}$ orthogonal to $\ket{\phi_1}$ and $\ket{\phi_2}$ and an angle $\alpha\in[0,\pi]$.

The angular distance $\theta_{13}$ between $\ket{\psi_1}$ and $\ket{\psi_3}$ then satisfies
\begin{align}\label{eq:cos2theta13}
\cos\theta_{13} &= \cos\theta_{12}\cos\theta_{23}+\sin\theta_{12}\sin\theta_{23}\cos\alpha.
\end{align}
In order to find the extremal values of $\theta_{13}$, notice that setting the derivative of the rhs with respect to $\alpha$ to zero yields the values $\cos(\theta_{12}-\theta_{23})=\cos(|\theta_{12}-\theta_{23}|)$ for $\alpha=0$ and $\cos(\theta_{12}+\theta_{23})$ for $\alpha=\pi$. Using $|\theta_{12}-\theta_{23}|<\theta_{12}+\theta_{23}$ and the fact that the cosine function is strictly decreasing on $[0,\pi]$, one obtains the lower and upper bounds of \cref{eq:triangle}. Moreover, since these extremal values are only reached for $\alpha=0,\pi$, saturation of \cref{eq:triangle} is only possible when all vectors belong to the plane spanned by $\ket{\phi_1}$ and $\ket{\phi_2}$.
\end{proof}

\begin{Lemma} \label{lemma:coplanar}
    Consider normalized vectors $\vec{m}_{\mathsf{a},x}, \vec{n}^{\a}_{x,y} \in \mathbb{R}^{12}$ for all $\mathsf{a} \in \{-1,1\}$, $x,y \in \{0,1\}$. Suppose that:
    \begin{enumerate}
        \item $\forall (\mathsf{s}, \mathsf{t})\in\{-1,1\}^2, \ \operatorname{asin} \cor{\vec{m}_{\mathsf{s},0}|\vec{n}^\s_{0,0}} +\operatorname{asin} \cor{\vec{m}_{\mathsf{t},1}|\vec{n}^\t_{1,0}} -\operatorname{asin} \cor{\vec{m}_{\mathsf{s},0}|\vec{n}^\s_{0,1}}+\operatorname{asin} \cor{\vec{m}_{\mathsf{t},1}|\vec{n}^\t_{1,1}} =\pi$ \\
        \item $\exists \lambda_x \in (0,1), l \in \mathbb{R}$ s.t.~$\lambda_x \cor{\vec{n}^+_{x,0}|\vec{n}^+_{x,1}} + (1-\lambda_x) \cor{\vec{n}^-_{x,0}|\vec{n}^-_{x,1}} = l$ (doesn't depend on $x$).
    \end{enumerate}
    Then all triples of vectors $\vec{m}_{\mathsf{a},x}$, $\vec{n}^\a_{x,0}$, $\vec{n}^\a_{x,1}$ lie in the same plane and $\cor{\vec{n}^\a_{x,0}|\vec{n}^\a_{x,1}}$ doesn't depend on $\mathsf{a},x$.
\end{Lemma}

\begin{proof}
    Let's consider the angular distances $\gamma_{x,y}^\a, \theta_{\mathsf{a},x} \in[0,\pi]$ defined as:
\begin{equation}
\cos(\gamma_{x,y}^\a) = \cor{\vec{m}_{\mathsf{a},x}|\vec{n}^\a_{x,y}}, \quad \cos(\theta_{\mathsf{a},x}) = \cor{\vec{n}^\a_{x,0}|\vec{n}^\a_{x,1}}.
\end{equation}
Using $\operatorname{asin} = \pi/2 - \operatorname{acos}$, the first assumption of the lemma then becomes:
\begin{equation} \label{eq:gamma_angles}
    \forall (\mathsf{s}, \mathsf{t})\in\{-1,1\}^2, \ \gamma^\s_{0,0} + \gamma^\t_{1,0} - \gamma^\s_{0,1} + \gamma^\t_{1,1} = 0
\end{equation}
Moreover, \cref{lemma:triangle} applied to $\vec{n}^\a_{x,0}$, $\vec{n}^\a_{x,1}$ and $\vec{m}_{\a,x}$ grants:
\begin{equation}
\label{eq:theta_angles}
    |\gamma_{x,0}^\a - \gamma_{x,1}^\a | \leq \theta_{\a,x} \leq \gamma_{x,0}^\a + \gamma_{x,1}^\a.
\end{equation}

Let us now consider the orderings of $\theta_{\a,x}$ which are compatible with the second requirement of the lemma. Note that this implies that there exist convex combinations of $\cos(\theta_{\a,x})$ for $\a=\pm$ that doesn't depend on $x$. In particular, this implies that it is impossible to have $\theta_{\t,1} < \theta_{\s,0}$ for all $\s,\t$. Without loss of generality, let's consider that $\theta_{+,0}\leq\theta_{+,1}$ (the other cases are treated similarly). We obtain from \cref{eq:gamma_angles,eq:theta_angles}:
\begin{equation}\label{eq:case1a}
0 \leq \gamma_{1,0}^+ + \gamma_{1,1}^+ = \gamma_{0,1}^+ - \gamma_{0,0}^+ \leq \theta_{+,0}\leq \theta_{+,1} \leq \gamma_{1,0}^+ + \gamma_{1,1}^+,
\end{equation}
which implies that $\theta_{+,0}=\theta_{+,1}$. The second hypothesis of the lemma then implies that either $\theta_{-,0}, \theta_{-,1} \leq \theta_{+,0}=\theta_{+,1}$ or $\theta_{-,0}, \theta_{-,1} \geq \theta_{+,0}=\theta_{+,1}$. Let's consider the first case, where we have $\theta_{-,0} \leq \theta_{+,1}$ (again, the other one is similar). We obtain from \cref{eq:gamma_angles,eq:theta_angles}:
\begin{equation}\label{eq:case1b}
0 \leq \gamma_{1,0}^+ + \gamma_{1,1}^+ = \gamma_{0,1}^- - \gamma_{0,0}^- \leq \theta_{-,0}\leq \theta_{+,1} \leq \gamma_{1,0}^+ + \gamma_{1,1}^+,
\end{equation}
and thus $\theta_{-,0}=\theta_{+,1}$. Since three of the thetas match, hypothesis 2 implies that all four are equal, and that
\begin{equation}\label{eq:case1c}
0 \leq \gamma_{1,0}^- + \gamma_{1,1}^- = \gamma_{0,1}^- - \gamma_{0,0}^- \leq \theta_{-,0}= \theta_{-,1} \leq \gamma_{1,0}^- + \gamma_{1,1}^-,
\end{equation}

The saturation of the triangle inequalities in \cref{eq:case1a,eq:case1b,eq:case1c} implies that $\vec{m}_{\mathsf{a},x}$, $\vec{n}_{0,\mathsf{a},x}$ and $\vec{n}_{1,\mathsf{a},x}$ must be co-planar for all $\a,x$. Furthermore, we obtained $\theta_{\a,x}=\theta_{a',x'}$ for all $\a,\a',x,x'$, i.e.~that $\cor{\vec{n}_{0,\mathsf{a},x}|\vec{n}_{1,\mathsf{a},x}}$ doesn't depend on $\a,x$.
\end{proof}

\subsection{Preliminary Quantum Considerations}

\begin{Lemma}{\cite[Theorem 2.5]{Chen23}} \label{lemma:extremeselftest}
    In the CHSH scenario, suppose that a correlation vector verifies the following conditions:
    \begin{enumerate}
        \item It is extremal in $\mathcal{Q}$,
        \item Up to local unitaries, there exists a unique pure, projective and non-degenerate realisation $(\ket{\phi_\theta}, A_x, B_y)$ of local dimension two giving those correlations.
    \end{enumerate}
    Then this correlation vector self-tests the realisation $(\ket{\phi_\theta}, A_x, B_y)$. 
\end{Lemma}
\medskip

\begin{Remark}
The fact that the correlation vector is supposed \textbf{extremal} is key to the previous Lemma, which of course does not hold if it's not the case.
\end{Remark}

\begin{Lemma}{\cite[Theorem 1]{Wang16}}\label{lemma:asinforphi+}
    Any correlation verifying $\cor{A_x} = \cor{B_y} = 0$ for all $x,y$,
    \begin{align}
        \operatorname{asin} \cor{A_0 B_0} +\operatorname{asin} \cor{A_1 B_0}-\operatorname{asin} \cor{A_0 B_1}+\operatorname{asin}\cor{A_1 B_1} =\pi
    \end{align}
    and such that at most one $\cor{A_xB_y}$ is equal to $\pm 1$, self-tests the maximally entangled two qubit state $\ket{\phi^+}$ along with alternating measurements on the $(\sigma_z, \sigma_x)$-plane.
\end{Lemma}
\medskip

\begin{Remark}
If one obtains correlations verifying the conditions of the lemma with a realization of local dimension~2, self-testing ensures the existence of unitaries $U_A, U_B \in L(\mathbb{C}^2)$ s.t: 
\begin{equation}
\begin{split}
    (U_A \otimes U_B) \ket{\psi} & = \ket{\phi^+}, \\
    U_A A_x U_A^\dagger & = \cos(a_x) \sigma_z + \sin(a_x) \sigma_x, \\
    U_B B_1 U_B^\dagger & = \cos(b)\sigma_z+ \sin(b) \sigma_x, \\
    \quad U_B B_0 U_B^\dagger &= \sigma_z.
\end{split} 
\end{equation}
where $0\leq a_0 \leq b \leq a_1 \leq \pi$ are fully determined by $\cor{A_x B_y}$. 
\end{Remark}

\begin{Lemma} \label{lemma:non-local}
    Consider a nonlocal correlation of the set $\mathcal{Q}$, then:
    \begin{enumerate}
        \item $\cor{A_x}\neq \pm 1$ and as such the correlators $[\Tilde{A}^{\a}_x B_y] = \frac{\cor{A_xB_y}+\a\cor{B_y}}{1+\a\cor{A_x}}$ are well-defined
        \item There is always at least two equations in the four equations of \cref{thm:self-test} in which the number of correlators $[\Tilde{A}^{\a}_x B_y]$ equal to $\pm 1$ is at most 1. Moreover, the number of correlators equal to  $\pm 1$ in the other two equations is at most 2.
    \end{enumerate}
\end{Lemma}
\begin{proof}
    The proof mainly relies on a result from Ref.~\cite{Chen23} where the authors prove that if a quantum correlation has two zero of probabilities for a fixed input choice $x$ of Alice, then it is local. In terms of correlations, it means that if two equations of the form $1+\a \cor{A_x}+\b \cor{B_y}+\a \b \cor{A_x B_y}=0$ are satisfied for the same $(\a,x)$ or the same $(\b,y)$, then the correlation is local. They also prove that a non-local point can have at most three zeros of positivities. 

    Therefore, one can verify that if $\cor{A_x}=\pm 1$, non-signaling condition imply that $\b(\cor{B_y}\mp\cor{A_x B_y}) \geq 0$ for all $\b$ and thus that $\cor{B_y}\mp\cor{A_x B_y}=0$. Then both equations $1\mp \cor{A_x}+\b \cor{B_y} \mp \b \cor{A_x B_y}=0$ for $\b=\pm1$ are satisfied, and the point is local. Conversely, if the point is non-local, then $\cor{A_x}\neq \pm 1$.

    For the second point, we have that $[\Tilde{A}^{\a}_x B_y]=\pm 1$ iff $1 +\a \cor{A_x} \mp \a \cor{B_y} \mp \cor{A_x B_y}=0$, and as such verifying $[\Tilde{A}^{\a}_x B_y]=\pm 1$ is equivalent to having a zero of probability for tuples $(\a,x)$ and $(\b=\mp \a, y)$. Let's consider a non-local point. It can have at most three zeros of probabilities. We suppose that this is the case (other cases with fewer zeros can be treated similarly). Since a non-local point cannot have two zeros of probability for a fix tuple $(\a,x)$, two zeros occur for a given $x$ and both values of $\a$ and one for the other input value $x'$ and a single output value $\a'$. Without loss of generality, let's consider that two zeros happen for $x=0$ and both $\a \in \{\pm 1\}$ and the last one for $x'=1$ and $\a'=+1$. Now, we obtain that $[\Tilde{A}^{+}_0 B_{y_{+,0}}]$, $[\Tilde{A}^{-}_0 B_{y_{-,0}}]$, $[\Tilde{A}^{+}_1 B_{y_{+,1}}]$ are equal to $\pm 1$ for some given inputs $y_{+,0}$, $y_{-,0}$, $y_{+,1}$. Considering the four equations of \cref{thm:self-test}, both equations for tuple $(\s,+)$ have two terms s.t. $[\Tilde{A}^\a_x B_y]=\pm 1$ ($[\Tilde{A}^{\s}_0 B_{y_{\s,0}}]$ and $[\Tilde{A}^{+}_1 B_{y_{+,1}}]$), and both equations for tuple $(\s,-)$ only have one ($[\Tilde{A}^{\s}_0 B_{y_{\s,0}}]$). Since no other zeros of probabilities can happen, no other terms $[\Tilde{A}^\a_x B_y]$ can be equal to $\pm1$. 
\end{proof}
\begin{Lemma} \label{lemma:uniqueness}
    For any pure, projective and non-degenerate realization of local dimension two $(\ket{\psi}, A_x, B_y)$  verifying the assumptions of \cref{thm:self-test}, there exists local unitaries $U_A, U_B \in L(\mathbb{C}^2) $ such that:
    \begin{equation}
        \begin{split}
            (U_A \otimes U_B) \ket{\psi} & = \ket{\phi_\theta}, \\
            U_A A_x U_A^\dagger & = \cos(a_x) \sigma_z + \sin(a_x) \sigma_x, \\
            U_B B_y U_B^\dagger & = \cos(b_y) \sigma_z + \sin(b_y) \sigma_x,
        \end{split}
    \end{equation}
    and the parameters $\theta, a_x, b_y$ are fully determined by the correlators $\cor{A_x}$, $\cor{B_y}$ and $\cor{A_x B_y}$. 
\end{Lemma}

\begin{proof}
    Up to local unitaries $U_A, U_B \in L(\mathbb{C}^2) $, one can map any realization $(\ket{\psi}, A_x, B_y)$ to:
    \begin{equation}
        \begin{split}
            \text{state: } & \ket{\phi_\theta}, \quad \theta \in [0,\pi/4], \\
            \text{measurements: } & A_x = \cos(a_x) \sigma_z + \sin(a_x) (\cos(\alpha_x) \sigma_x + \sin(\alpha_x) \sigma_y), \\
            & B_0 = \cos(b_0) \sigma_z + \sin(b_0) (\cos(\beta) \sigma_x + \sin(\beta) \sigma_y), \\
            & B_1 = \cos(b_1) \sigma_z + \sin(b_1) \sigma_x.
        \end{split}
    \end{equation}

The case where the marginals $\cor{A_x} = \cor{B_y} = 0$ vanish is due to the self-testing result \cite{Wang16} (see \cref{lemma:asinforphi+}) so we can exclude the zero marginals case. Since the correlations are supposed non-local, we have $\theta \neq 0,\pi/4$. One can apply the steering transformation introduced in~\cref{sec:steering} to ensure that all $[\Tilde{A}^{\a}_x B_y]$ are correlators measured on a singlet state $\ket{\phi^+}$, where:
\begin{equation}
        \begin{split}
            & \Tilde{A}^{\a}_x = \cos(\Tilde{a}^{\a}_x) \sigma_z + \sin(\Tilde{a}^{\a}_x) (\cos(\alpha_x) \sigma_x + \sin(\alpha_x) \sigma_y), \\
            & B_0 = \cos(b_0) \sigma_z + \sin(b_0) (\cos(\beta) \sigma_x + \sin(\beta) \sigma_y), \\
            & B_1 = \cos(b_1) \sigma_z + \sin(b_1) \sigma_x.
        \end{split}
\end{equation}
Note that since the correlations are non-local, we have $a_0 \ncong a_1 [\pi]$ and the fact that the state is not maximally entangled ensures that either $\Tilde{a}^{+}_{0}\neq\Tilde{a}^{-}_{0}$ or $\Tilde{a}^{+}_{1}\neq\Tilde{a}^{-}_{1}$. 

The fact that correlations are taken non-local ensures by \cref{lemma:non-local} that one can apply the self-testing result \cref{lemma:asinforphi+} for at least two pairs of $(\s,\t)$, certifying the angle of measurements $B_0, B_1$ with respect to $(\Tilde{A}_{\s,0},\Tilde{A}_{\s,0})$ in those cases. Note that if \cref{lemma:asinforphi+} doesn't apply for the other pairs, it means that some correlators are equal to $\pm 1$, i.e.~$[\Tilde{A}^{\a}_xB_y] = \pm 1$. Since the vector $\Tilde{A}^{\a}_x\ket{\phi^+}$ and $B_y \ket{\phi^+}$ are unit and $\ket{\phi^+}$ is full rank, this would mean that $\Tilde{A}^{\a}_x=B_y$, fully identifying the missing measurements. Finally, since the initial state is already $\ket{\phi^+}$, there exists $U\in L(\mathbb{C}^2)$ s.t:
\begin{equation} \label{eq:basismeas}
\begin{split}
    & \Bar{U} \Tilde{A}^{\a}_x \Bar{U}^\dagger = \cos(w_{\a,x}) \sigma_z + \sin(w_{\a,x}) \sigma_x, \\
    & U B_0 U^\dagger = \cos(b) \sigma_z + \sin(b) \sigma_x, \quad U B_1 U^\dagger = \sigma_z 
\end{split}
\end{equation}
where $0\leq w_{\s,0} \leq b \leq w_{\t,1} \leq \pi$ for all $\s, \t$ are determined by the correlators $[\Tilde{A}^{\a}_xB_y]$. 

In particular, we obtain that all vectors $\Tilde{A}^{\a}_x\ket{\phi^+}$, $\Bar{B_y\ket{\phi^+}}$ lie in the same plane. If one consider $x$ such that $\Tilde{A}^{+}_{x}\neq\Tilde{A}^{-}_{x}$, then $\Tilde{A}^{+}_{x}\ket{\phi^+}$ and $\Tilde{A}^{-}_{x}\ket{\phi^+}$ form a plane of angle $\alpha_x$, with respect to $\sigma_x - \sigma_y$. This implies that it is always possible to perform unitaries $V_A = \operatorname{diag}(1, e^{-i\alpha_x})$ and $V_B = V_A^\dagger$ such that measurements $A_x, B_y$ are real and the state is unchanged ($(V_A\otimes V_B) \ket{\phi_\theta}= \ket{\phi_\theta}$).

We are now left with only real matrices and thus we can parametrize $U$ as a real unitary matrix of the form:
\begin{equation}
    U = \begin{pmatrix}
        \cos(\gamma) & - \epsilon_u \sin(\gamma)\\
        \sin(\gamma) & \epsilon_u \cos(\gamma) \\
    \end{pmatrix}
\end{equation}
where $\epsilon_u \in \{-1,1\}$ and $\gamma \in [0,2\pi)$. Since a global sign flip doesn't change the measurement, one can assume $\gamma \in [-\pi/2,\pi/2]$. Likewise, the sign $\epsilon_u$ can be absorbed by applying a $\sigma_z$ on both parties (which preserves the state $\ket{\phi_\theta}$). Therefore, for the rest of the proof, we assume that $U$ is a rotation of angle $\gamma$.

Let us write the initial measurements $A_x = \ketbra{u_{+,x}}-\ketbra{u_{-,x}}$ and denote $\ket{v_{\a,x}}:= T_\theta(\ket{u_{\a,x}})$. Since the vectors $\ket{v_{\a,x}}$ correspond to the eigenvectors of $\Tilde{A}^{\a}_x$ of eigenvalue $\a$, they are fixed by \cref{eq:basismeas}:
\begin{equation}
    \begin{split}
        & \ket{v_{+,x}} = U^\dagger (\cos(\frac{w_{+,x}}{2}) \ket{0} + \sin(\frac{w_{+,x}}{2}) \ket{1}):=U^\dagger \ket{W_{+,x}}, \\
        & \ket{v_{-,x}} = U^\dagger (\sin(\frac{w_{-,x}}{2}) \ket{0} - \cos(\frac{w_{-,x}}{2}) \ket{1}):= U^\dagger \ket{W_{-,x}}.
    \end{split}
\end{equation}
Note that in~\cref{sec:steering}, we showed that $\ket{u_{\a,x}}= T_{\pi/2-\theta}(\ket{v_{\a,x}})$. As such the measurements $A_x$ are fully determined by the knowledge of $\theta$ and $\gamma$. More precisely:
\begin{equation}
    \ket{u_{\a,x}} \propto (\sin(\theta)\cor{0|U^\dagger|W_{\a,x}}\ket{0} + \cos(\theta)\cor{1|U^\dagger|W_{\a,x}} \ket{1}).
\end{equation}
where the normalization constants are given by $\braket{u_{\a,x}}=1$. Since the original measurements $A_x$ are unitary, we have $\cor{u_{+,x}|u_{-,x}}=0$ which implies that for all $x$:
\begin{equation}
    \sin^2(\theta)\cor{0|U^\dagger|W_{+,x}}\cor{0|U^\dagger|W_{-,x}} + \cos^2(\theta)\cor{1|U^\dagger|W_{+,x}}\cor{1|U^\dagger|W_{-,x}} = 0.
\end{equation}
The lhs can be simplified by introducing $\gamma_x = \gamma - (w_{+,x}+w_{-,x})/4$, leading to
\begin{equation} \label{eq:thetagamma}
    \frac{1}{2}\left[\cos(2\theta)\cos(2\gamma_x)-\sin(\frac{w_{+,x}+w_{-,x}}{2})\right]=0.
\end{equation}
By linear combination of those equations, one obtains
\begin{equation}
    \tan(\gamma) = \frac{\sin(\frac{w_{+,0}-w_{-,0}}{2})\cos(\frac{w_{+,1}+w_{-,1}}{2})-\sin(\frac{w_{+,1}-w_{-,1}}{2})\cos(\frac{w_{+,0}+w_{-,0}}{2})}{\sin(\frac{w_{+,0}-w_{-,0}}{2})\sin(\frac{w_{+,1}+w_{-,1}}{2})-\sin(\frac{w_{+,1}-w_{-,1}}{2})\sin(\frac{w_{+,0}+w_{-,0}}{2})},
\end{equation}
where we allow both sides to be infinite. This equation fixes a unique $\gamma\in[-\pi/2,\pi/2]$ and \cref{eq:thetagamma} fixes $\theta$ to be either $\theta^\star$ or $\pi-\theta^\star$, and only one of those two solutions can belong to $(0,\pi/4)$.

\end{proof}

\subsection{\label{sec:level1} Proof of \cref{thm:self-test}}

Let us consider non-local correlations in the CHSH scenario satisfying \cref{eq:conditions}. \cref{prop:realqubitreduction} reduces the search for extremal points to the ones admitting a pure real realization of local dimension 2 and since every correlation is a convex combination of extremal correlations, one can write:
\begin{equation}
    \forall x,y \in \{-1,0,1\}, \ \cor{A_x B_y} = \sum_{i=1}^N p_i \cor{A_x^{(i)} B_y^{(i)}}
\end{equation}
where $p_i \geq 0$, $\sum_i p_i =1$, $A_{-1}=B_{-1}=\id$ and on each block $i$ the correlation vector $\{\cor{A_x^{(i)} B_y^{(i)}}\}_{x,y}$ is in $\mathcal{Q}_2$. To prove that the initial correlations are extremal in $\mathcal{Q}$, we need to prove that the sub-correlations $\{\cor{A_x^{(i)} B_y^{(i)}}\}_{x,y}$ don't depend on the register $i$.

Such a decomposition allows considering an overall underlying realization with local Hilbert spaces encoding a qubit space and a classical register: $\mathcal{H}_{A(B)}= \mathcal{R}^2 \otimes \mathbf{N}$ (where $\mathbf{N}$ is the set of natural numbers from $1$ to $N$). The overall state and measurements are:
\begin{equation}
    \begin{split}
        \text{state: }& \ket{\psi} = \sum_{i=1}^N p_i \ket{\psi^i} \otimes \ket{ii},\\
        \text{measurements: }& A_x = \sum_{i=1}^N A_x^{(i)} \otimes \ketbra{i}, \quad B_y = \sum_{i=1}^N B_y^{(i)} \otimes \ketbra{i}.
    \end{split}
\end{equation}
For each $i$ one can find unitaries $U_A^i$, $U_B^i$ such that the underlying state is $\ket{\phi_{\theta^i}}$, with $\theta^i \in [0,\pi)$. Considering the local unitaries $U_{A(B)}=\sum_i U_{A(B)}^i \otimes \ketbra{i}$, the overall state can then be written
\begin{equation} 
    (U_A\otimes U_B) \ket{\psi} = \sum_i p_i \ket{\phi_{\theta^i}} \otimes \ket{ii}.
\end{equation}

Let us now consider the steering transformation $T_{\theta^i}$ introduced in~\cref{sec:steering} on each block. We know that it allows to consider new operators $\Tilde{A}^{\a\, (i)}_x \in L(\mathbb{R}^2)$ (see \cref{eq:newoperatorsdef}) such that
\begin{equation} \label{eq:blockrelation}
    \cor{ \phi_{\theta^i} | A_x^{(i)} B_y^{(i)} | \phi_{\theta^i}} + \a \cor{ \phi_{\theta^i} | B_y^{(i)} | \phi_{\theta^i}} = ({ 1+\a \cor{ \phi_{\theta^i} | A_x^{(i)} | \phi_{\theta^i}}})\bra{\phi^+} \Tilde{A}^{\a\, (i)}_x B_y^{(i)} \ket{\phi^+}, 
\end{equation}
Let us now consider that
\begin{equation}
    1+\a \cor{A_x} = 1 + \a \sum_{i=1}^N p_i  \cor{ \phi_{\theta^i} | A_x^{(i)} | \phi_{\theta^i}} =  \sum_{i=1}^N p_i ({ 1+\a \cor{ \phi_{\theta^i} | A_x^{(i)} | \phi_{\theta^i}}}) .
\end{equation}
Since the overall correlation is non-local, \cref{lemma:non-local} ensures that $1+\a \cor{A_x} \neq 0$. One can therefore define
\begin{equation}
    \alpha_{\a,x}^i = \frac{p_i ({ 1+\a \cor{ \phi_{\theta^i} | A_x^{(i)} | \phi_{\theta^i}}}) }{1+\a \cor{A_x}}
\end{equation}
which verifies $\alpha_{\a,x}^i \geq 0$ and $\sum_i \alpha_{\a,x}^i = 1$. Now using \cref{eq:blockrelation} we have:
\begin{equation}
\begin{split}
    \frac{\cor{A_x B_y} + \a \cor{B_y}}{1+\a \cor{A_x}} & = \sum_{i=1}^N \alpha_{\a,x}^i \bra{\phi^+} \Tilde{A}^{\a\, (i)}_x B_y^{(i)} \ket{\phi^+}. 
\end{split}
\end{equation}

We now define a four new overall states and measurements operators as
\begin{equation}\label{eq:newmeasurements}
\begin{split}
\ket{\Phi^+_{\a,x}}&=\sum_{i}\sqrt{\alpha_{\a,x}^i} \ket{\phi^+} \otimes \ket{ii}, \quad  \Tilde{A}^{\a}_x = \sum_{i} \Tilde{A}^{\a}_{x^{i}} \otimes \ketbra{i}, \quad B_y = \sum_i B_y^{(i)} \otimes \ketbra{i}. \\
\end{split}
\end{equation}
Denoting expectations values on $\ket{\Phi^+_{\a,x}}$ as $\langle\cdot\rangle_{\a,x}$, we obtain:
\begin{equation}
\begin{split}
    & \langle \Tilde{A}^{\a}_x \rangle_{\a,x} = 0, \langle B_y \rangle_{\a,x} = 0, \\
    & \frac{\cor{A_x B_y} + \a \cor{B_y}}{1+\a \cor{A_x}}  = \cor{\Tilde{A}^{\a}_x B_y}_{\a,x}. 
\end{split}
\end{equation}
Note that the marginals are always $0$ because for every classical index $i$ the measured state $\ket{\phi^+}$ is maximally entangled.

Since the measurement operators in \cref{eq:newmeasurements} are real and unitary, and the state $\ket{\Phi_{\a,x}^+}$ is real and normalized, all vectors $\Tilde{A}^{\a}_x \ket{\Phi_{\a,x}^+}$, $B_y \ket{\Phi_{\a,x}^+}$ are real and normalized. Notice that we have:
\begin{subequations}
    \begin{align}
        & \langle B_0 B_1 \rangle_{\a,x}=\sum_{i}\alpha_{\a,x}^i \bra{\phi^+}B_0^i B_1^i\ket{\phi^+}, \\
        & (1+\cor{A_x})\alpha_{+,x}^i +  (1-\cor{A_x})\alpha_{-,x}^i = 2p_i
    \end{align}
\end{subequations}
Therefore, with $\lambda_x = (1+\cor{A_x})/2 \in (0,1)$ and utilizing the assumption of \cref{thm:self-test}, we obtain:\begin{enumerate}
    \item $\forall (\s,\t) \in\{-1,1\}^2, \ \operatorname{asin} \cor{\Tilde A_0^\s B_0}_{\s,0} +\operatorname{asin} \cor{\Tilde A_1^\t B_0}_{\t,1} -\operatorname{asin} \cor{\Tilde A_0^\s B_1}_{\s,0}+\operatorname{asin}\cor{\Tilde A_1^\t B_1}_{\t,1} =\pi$ \\
    \item $\lambda_x \langle B_0B_1 \rangle_{+,x} + (1-\lambda_x)\langle B_0B_1 \rangle_{-,x} = \sum_{i} p_i \bra{\phi^+}B_0^i B_1^i\ket{\phi^+}$ doesn't depend on $x$. 
\end{enumerate}
We can thus use \cref{lemma:coplanar} with $\vec{m}_{\a,x}= \Tilde A_x^\a \ket{\Phi_{\a,x}^+}$ and $\vec{n}^\a_{x,y}= B_y\ket{\Phi_{\a,x}^+}$ to obtain that the following triples of vectors must be coplanar for all $\a,x$:
\begin{align}
B_0\ket{\Phi_{\a,x}^+}, B_1\ket{\Phi_{\a,x}^+}, \Tilde A_x^\a \ket{\Phi_{\a,x}^+}.
\end{align}
and that $\langle B_0B_1 \rangle_{\a,x}:=C$ doesn't depend on $\a,x$. 

This allows us to write for all $\a,x$: 
\begin{equation}
    \Tilde{A}^{\a}_x\ket{\Phi^+_{\a,x}} = r_{\a,x}^0 B_0\ket{\Phi^+_{\a,x}} + r_{\a,x}^1 B_1\ket{\Phi^+_{\a,x}} 
\end{equation}
where $r_{\a,x}^y \in \mathbb{R}$. By projecting on each classical register $\ketbra{i}$, we get
\begin{equation}
    (\Tilde{A}^{\a}_x)^{i}\ket{\phi^+}= r_{\a,x}^0 B_0^{i}\ket{\phi^+} +  r_{\a,x}^1 B_1^{i}\ket{\phi^+} 
\end{equation}
Now we can sum over all classical registers but with weights $\sqrt{\alpha_{\a',x'}^i}$ to obtain
\begin{equation}
    \Tilde{A}^{\a}_x\ket{\Phi^+_{\a',x'}}= r_{\a,x}^0 B_0\ket{\Phi^+_{\a',x'}} + r_{\a,x}^1  B_1\ket{\Phi^+_{\a',x'}} 
\end{equation}
And finally we can compute
\begin{equation}
    \langle \Tilde{A}^{\a}_x B_y \rangle_{\a',x'} = r_{\a,x}^y \langle B_y^2 \rangle_{\a',x'}  + r_{\a,x}^0 r_{\a,x}^1 \langle B_0 B_1\rangle_{\a',x'}  = r_{\a,x}^y + r_{\a,x}^0 r_{\a,x}^1 C
\end{equation}
Since the right-hand term doesn't depend on $\a',x'$ neither can the left-hand one. Thus, we have proven that $\langle A_x ^\a B_y \rangle_{\a',x'}$ doesn't depend on the indexes $\a',x'$. Notably, we can rewrite the assumption of \cref{thm:self-test} for a single arbitrary state $\ket{\Phi^+}:= \ket{\Phi^+_{+,0}}$:
\begin{equation} \label{eq:samephi+}
    \forall (\s,\t)\in\{-1,1\}^2, \ \operatorname{asin} \cor{\Tilde A_0^\s B_0}_{\Phi^+} +\operatorname{asin} \cor{\Tilde A_1^\t B_0}_{\Phi^+} -\operatorname{asin} \cor{\Tilde A_0^\s B_1}_{\Phi^+}+\operatorname{asin}\cor{\Tilde A_1^\t B_1}_{\Phi^+} =\pi
\end{equation}

For all $(\s,\t)$, we are now dealing with a realization on the state $\ket{\Phi^+}$, with two measurements $\Tilde A_0^\s$, $\Tilde A_1^\t$ and $B_0$, $B_1$ for each party. The corresponding quantum correlations can be decomposed as a convex mixture of $N$ sub-correlations due to the classical register encoded in $\ket{\Phi^+}$. Since it has zero marginals and verifies \cref{eq:samephi+}, the work of \cite{Masanes03} ensures that the overall correlation is extremal for all $\s,\t$ and thus all sub-correlations are equals. Therefore:
\begin{equation}
    \{\bra{\phi^+} \Tilde{A}^{\a\, (i)}_x B_y^{(i)} \ket{\phi^+}\}_{\a,x,y} \ \text{doesn't depend on the register $i$.}
\end{equation}
Moreover, it means that every initial sub-correlations, obtained by measuring state $\ket{\phi_{\theta^i}}$ with real measurements $(A_x^{(i)}, B_y^{(i)})$, satisfy the condition of \cref{thm:self-test}. Since \cref{lemma:uniqueness} ensures that such a realization is determined by the values of $\bra{\phi^+} \Tilde{A}^{\a\, (i)}_x B_y^{(i)} \ket{\phi^+}$, we obtain that $\theta^i$ and $(A_x^{(i)}, B_y^{(i)})$ don't depend on $i$. This implies that the sub-correlations $\{\cor{A_x^{(i)} B_y^{(i)}}\}_{x,y}$ don't depend on $i$ and thus the extremality of the initial correlations. The self-testing part is then obtained combining extremality and \cref{lemma:uniqueness} together with \cref{lemma:extremeselftest}.

\section{Proof of \cref{thm:non-exp}} \label{sec:proofLemma2}

The correlation distribution corresponding to a realization \cref{eq:qubitrealization} in $\mathcal{Q}_2$ is:
\begin{equation} \label{eq:point}
    \vec P_{\theta,a_x,b_y} = \begin{array}{c|c|c}
          & \cos({2\theta})\cos({b_0}) & \cos({2\theta})\cos({b_1}) \\
         \hline
         \cos({2\theta})\cos({a_0}) & \multicolumn{2}{c}{\multirow{2}{*}{$\cos({a_x})\cos({b_y})+\sin({2\theta})\sin({a_x})\sin({b_y})$}} \\
         \cmidrule{1-1}
         \cos({2\theta})\cos({a_1}) & \multicolumn{2}{c}{}
    \end{array}, \quad \theta,a_x,b_y \in \mathbb{R}
\end{equation}
While \cref{prop:qubitrealization} guarantees that every extremal point in $\mathcal{Q}$ is of this form, it is not granted that every such point is extremal in $\mathcal{Q}$. In the following, we express a condition on the parameters $\theta, a_x, b_y$ for such a point to be non-exposed, i.e.~not to be the unique maximizer of any Bell expression. 

In all generality, a Bell expression in the CHSH scenario can be denoted by a real vector $\Vec{\beta} \in \mathbb{R}^8$. The value of the Bell expression on $\Vec P$ is given by the scalar product $\Vec \beta \cdot \Vec P$ and thus $\Vec P$ is non-exposed in $\mathcal{Q}$ iff
\begin{equation} \label{eq:nonexp}
    \forall \Vec \beta \in \mathbb{R}^8, \ \exists \Vec {P'} \in \mathcal{Q}, \ \text{s.t} \ \Vec \beta \cdot \Vec{P'} \geq \Vec \beta \cdot \Vec{P},\ \Vec{P}'\neq\Vec{P}
\end{equation}
If we denote by $\mathcal{C}_P$ the set of all Bell expressions for which $\Vec P$ reaches their maximal quantum value, i.e.~$\mathcal{C}_P~=~\{\beta\in\mathbb{R}^8 | \Vec \beta \cdot \Vec{P} = \max_{P'\in\mathcal{Q}} \Vec \beta \cdot \Vec{P'}\}$. Then the condition for $\Vec P$ to be non-exposed reduces to
\begin{equation} \label{eq:nonexp2}
    \forall \Vec \beta \in \mathcal{C}_P, \ \exists \Vec {P'} \in \mathcal{Q}, \ \text{s.t} \ \Vec \beta \cdot \Vec{P'} = \Vec \beta \cdot \Vec{P},\ \Vec{P}'\neq\Vec{P}.
\end{equation}

The proof of Theorem 3 is divided in two main parts, which are developped in the sections below. In the first one, we find necessary conditions on Bell expressions to be in the subspace $\mathcal{C}_P$. In the second one, we identify a point verifying \cref{eq:nonexp2} for every Bell expression in this subset. In what follows, we fix a choice of parameters $\theta, a_x, b_y$ in the region \cref{eq:range} and denote by $\Vec P$ the distribution $\Vec P_{\theta, a_x, b_y}$. Furthermore, since all realizations with $\theta=0$ only give local correlations, we assume $\theta >0$.

\subsection{Identification of Bell expressions maximized by $\vec P$}
For a given choice of measurements we introduce the measurement vector
\begin{equation}
    \Vec M = \{ A_0, A_1, B_0, B_1, A_0 B_0, A_1 B_0, A_0 B_1, A_1 B_1\} \in L(\mathcal{H}_A\otimes\mathcal{H}_B)^8.
\end{equation}
For every Bell expression $\beta$ with $\{\beta\}_{-1,-1}=0$, one can construct the Bell operator associated with this measurement choice as $S = \Vec \beta \cdot \Vec M$, which is an hermitian operator. For any state $\ket{\psi}$, the correlation distribution is given by $\vec {P'} = \bra{\psi}\Vec M \ket{\psi}$ and the value of the Bell expression is
\begin{equation}
    \Vec \beta \cdot \vec {P'} = \bra{\psi} S \ket \psi.
\end{equation}
As such, if the point $\vec P$ is to give the maximal quantum value of $\beta$, then the state $\ket{\phi_\theta}$ must be an eigenstate of $S$ (of maximal eigenvalue). This implies that for all vector $\ket{\psi^\bot}$ orthogonal to $\ket{\phi_\theta}$, $\ket{\psi^\bot}$ is also orthogonal to $S \ket{\phi_\theta}$, and as such
\begin{equation} \label{eq:eigcond}
    0 = \bra{\psi^\bot} S \ket{\phi_\theta}  = \Vec \beta \cdot \bra{\psi^\bot}  \Vec M \ket{\phi_\theta} = \Vec \beta \cdot \Vec T_{\ket{\psi^\bot}}
\end{equation}
where we denoted $\Vec T_{\ket{\psi^\bot}} = \bra{\psi^\bot}  \Vec M \ket{\phi_\theta}$. Note that this condition was recently used for solving optimization problems and prove that specific quantum points are non-exposed, see \cite{Goh18,Chen23}. 

Another condition to get a maximal violation is that for any small variation of the parameters $\theta, a_x, b_y$, the value of the Bell expression should be non-increasing. At first order, this gives
\begin{equation} \label{eq:varcond}
\begin{split}
    0 = \vec \beta \cdot \frac{\partial \vec P}{\partial \theta}= \vec \beta \cdot \frac{\partial \vec P}{\partial a_x} = \vec \beta \cdot \frac{\partial \vec P}{\partial b_y}
\end{split}
\end{equation}

We now consider a (possible) subset of the necessary conditions \cref{eq:eigcond} and \cref{eq:varcond} by considering only the three orthogonal states $\ket{\psi_\theta}=\sin(\theta) \ket{00} - \cos(\theta)\ket{11}$,  $\ket{01}$, $\ket{10}$, and variations along measurement angles $a_0$ and $b_0$. This allows us to say that if $\Vec P$ maximizes $\Vec \beta$, then 
\begin{equation}
    \Vec \beta \in V^\bot, \quad \text{where} \ V = \text{Vect}\left\langle \Vec T_{\ket{\psi_\theta}}, \Vec T_{\ket{01}}, \Vec T_{\ket{10}}, \frac{\partial \vec P}{\partial a_0}, \frac{\partial \vec P}{\partial b_0} \right\rangle 
\end{equation}
This allows us to conclude that $\mathcal{C}_P \subset V^\bot$ and to reduce the range of Bell expressions to this linear subspace for the rest of the argument.

\subsection{Non-exposed sufficient condition}
Let's suppose there exists a vector $\vec v \in V$ and a local vector $\vec L \in \mathcal{Q}$ such that $\vec L = \vec P + \vec v$, then 
\begin{equation}
    \forall \beta \in V^\bot, \vec \beta \cdot \vec L = \vec \beta \cdot \vec P
\end{equation}
Therefore if such a decomposition exists, and $\vec L \neq \vec P$, condition \cref{eq:nonexp} is satisfied and the point $\vec P$ is non-exposed in $\mathcal{Q}$. A sufficient condition to our problem is therefore to find a decomposition of the form
\begin{equation}\label{eq:decomp}
    \Vec L = \Vec{P} + \Vec{v}, \ \Vec L \in \mathcal{Q}, \ \Vec v \in V-\{0\}.
\end{equation}

We further restrict the choice of $\Vec L$ to be included in one of the following four subspaces:
\begin{equation}
    \forall \s, \t \in \{-1,+1\}, \ \mathcal{L}_{\s\t} =  \left\{ \begin{array}{c|c|c}
          & \alpha_0 & \alpha_1 \\
         \hline
         \s & \s\, \alpha_0 & \s\, \alpha_1 \\
         \hline 
         \t & \t\, \alpha_0 & \t\, \alpha_1 \\\end{array}, \alpha_0, \alpha_1 \in \mathbb{R} \right\}
\end{equation}
These spaces have nice properties. First, one can verify that the positivity constraints imply that $\alpha_0, \alpha_1 \in [-1,1]$, and that under these conditions, the value of all variants of the CHSH expression is upper bounded by 2. Therefore, the local, quantum and non-signaling sets coincide in these subspaces, and we can express the condition for a point $\Vec L \in \mathcal{L}_{\s\t}$ to be local (or quantum) simply as $\alpha_0, \alpha_1 \in [-1,1]$, or equivalently as $1-\alpha_x^2 \geq 0$ for all $x$. Second, due to the fact that $\theta >0$, we have $|\langle A_x\rangle|<1$ for all $x$ and thus $\Vec P \notin \mathcal{L}_{\s\t}$ for all $\s, \t$. This means that finding a decomposition of the form \cref{eq:decomp} for $\vec L \in \mathcal{L}_{\s\t}$ always ensures that $\vec v \neq 0$.\\

As any point in $V$ can be written as the linear combination 
\begin{equation}
    \Vec v_{x,y,z,a,b} = x \Vec T_{\ket{\psi_\theta}} + y \Vec T_{\ket{01}} + z \Vec T_{\ket{10}} + a \frac{\partial \vec P}{\partial a_0} + b \frac{\partial \vec P}{\partial b_0}
\end{equation}
where $x,y,z,a,b$ are five real parameters, finding a decomposition $\Vec L = \Vec{P} + \vec{v}$ where $\vec L \in \mathcal{L}_{\s\t}$ can be translated to the following linear system 
\begin{equation} \label{eq:system}
    \left\{ \begin{split}
        & \{\Vec P+\Vec v_{x,y,z,a,b}\}_{0,-1} = \s \\
        & \{\Vec P+\Vec v_{x,y,z,a,b}\}_{1,-1} = \t \\
        & \{\Vec P+\Vec v_{x,y,z,a,b}\}_{-1,0} = \s \{\Vec P+\Vec v_{x,y,z,a,b}\}_{0,0} = \t \{\Vec P+\Vec v_{x,y,z,a,b}\}_{1,0}\\
        & \{\Vec P+\Vec v_{x,y,z,a,b}\}_{-1,1} = \s \{\Vec P+\Vec v_{x,y,z,a,b})_{0,1} = \t \{\Vec P+\Vec v_{x,y,z,a,b}\}_{1,1}\\
    \end{split}\right.,
\end{equation}
where we recall that ${\Vec P}_{x,y}=\langle A_x B_y \rangle$ for $x,y\in \{-1,0,1\}$ with the convention that $A_{-1}=\id_{\mathcal{H}_A}$, $B_{-1}=\id_{\mathcal{H}_B}$.

We first focus on the case where $\theta \leq \pi/4$. In this case, the linear system admits a solution for any parameters verifying \cref{eq:range} and $(\s,\t)\neq (-1,1)$, given by:
\begin{subequations}
    \begin{align}
        & x = -\cos\left(\frac{a_0+\frac{1-\s}{2}\pi+(a_1+\frac{1-\t}{2}\pi)}{2}\right) \sin({2\theta}) /D_{\s\t}, \\
        & y = \sin\left(\frac{a_0+\frac{1-\s}{2}\pi}{2}\right) \cos\left(\frac{a_1+\frac{1-\t}{2}\pi}{2}\right) \sin({\theta}) /D_{\s\t}, \\
        & z = \cos\left(\frac{a_0+\frac{1-\s}{2}\pi}{2}\right) \sin\left(\frac{a_1+\frac{1-\t}{2}\pi}{2}\right) \cos({\theta}) /D_{\s\t}, \\
        & a = -\sin\left(\frac{a_0+\frac{1-\s}{2}\pi-(a_1+\frac{1-\t}{2}\pi)}{2}\right)/D_{\s\t}, \\
        & b = 0,
    \end{align}
\end{subequations}
where: 
\begin{equation}
    D_{\s\t} = \cos\left(\frac{a_0+\frac{1-\s}{2}\pi-(a_1+\frac{1-\t}{2}\pi)}{2}\right) + 
\cos\left(\frac{a_0+\frac{1-\s}{2}\pi+(a_1+\frac{1-\t}{2}\pi)}{2}\right) \cos({2\theta}).
\end{equation}
This solution gives a decomposition of the form $\Vec L = \Vec P + \vec v$ where $\Vec{L}\in \mathcal{L}_{\s\t}$ is given by the two parameters:
\begin{equation}
    \begin{split}
        \alpha_y^{\s\t} = \frac{\left(\left(\s \cos(a_0) + \cos(2\theta)\right) \sin(b_y) - 
   \s \cos(b_y) \sin(a_0) \sin(2\theta)\right) \left(\left(\t \cos(a_1) + \cos(2\theta)\right) \sin(b_y) - 
   \t \cos(b_y) \sin(a_1) \sin(2\theta)\right)}{D_{\s\t}}.
    \end{split}
\end{equation}
Note that for the parameters we chose here ($0 < \theta\leq\pi/4$, $0\leq a_0\leq a_1<\pi$) the denominator $D_{\s\t}$  is never $0$. 

To verify \cref{eq:decomp}, we now need to look at when the point $\Vec L$ belongs to the quantum set. As we said, this is equivalent to ask that their exist some $(\s,\t)$ such that $-1\leq \alpha_y^{\s\t} \leq 1$ for $y\in\{0,1\}$, or equivalently that
\begin{equation}
    \forall y \in \{0,1\}, \ 1-(\alpha_y^{\s\t})^2 \geq 0.
\end{equation}
The validity of this inequality is unchanged by multiplication with a positive scalar and as such one can look at the sign of
\begin{equation}
    \Delta_{\s\t} = \frac{D_{\s\t}^2}{(1+\s\langle A_0 \rangle)(1+\t\langle A_1\rangle)}(1-(\alpha_y^{\s\t})^2).
\end{equation}
This quantity can be written simply as
\begin{equation}
\Delta_{\s\t} = \s\, \t \, \sin(\Tilde{a}^{\s}_{0}-b_y)\sin(\Tilde{a}^{\t}_{1}-b_y).
\end{equation}
by considering the following change of variables:
\begin{equation}
\begin{split}
    a_0 \longrightarrow \Tilde{a}^{\s}_{0}, \quad a_1 \longrightarrow \Tilde{a}^{\t}_{1},
\end{split}
\end{equation}
where $\Tilde{a}^{\a}_x$ are defined in \cref{eq:anglemodif}. Thus, the positivity conditions for $\Vec L$ to be local now become equivalent to
\begin{equation} \label{eq:posaltcond}
    \forall y\in\{0,1\},\ \Delta_{\s\t} = \s\, \t \, \sin(\Tilde{a}^{\s}_{0}-b_y)\sin(\Tilde{a}^{\t}_{1}-b_y)\geq 0.
\end{equation}

We now think by contradiction and look for conditions that ensure that there is no choice of $(\s,\t)$ such that the above is verified:
\begin{enumerate}
    \item The choice of parameters that we made allows us to state that $b_y \leq a_0 \leq \Tilde{a}^{+}_{0}$ for all $y$. Therefore, $\sin(\Tilde{a}^{+}_{0} - b_y) \leq 0$ and the above condition for $(\s,\t)=(1,1)$ is verified whenever $\Tilde{a}^{+}_{1} \leq b_0$. We thus impose $\Tilde{a}^{+}_{1} \geq b_0$. \\
    \item Now, the condition for $(\s,\t)=(1,-1)$ is verified whenever $\sin(\Tilde{a}^{-}_{1} - b_y) \geq 0$ for both $y$. As $\theta \leq \pi/4$, $\Tilde{a}^{-}_{1} \geq \Tilde{a}^{+}_{1} \geq 0$ and thus it is always true for $y=0$. Therefore, we need to have $\Tilde{a}^{-}_{1}\leq b_1$.\\
    \item Last, the condition for $(\s,\t)=(-1,-1)$ and $y=1$ is now verified as $\Tilde{a}^{-}_{0} \leq \Tilde{a}^{-}_{1} \leq b_1$. Then the inequality for $y=0$ does not hold only when $\Tilde{a}^{-}_{0} \leq b_0$. \\
\end{enumerate}
Finally, we can conclude that none of the conditions \cref{eq:posaltcond} is verified only when the parameters verify:
\begin{equation} \label{eq:fullalt}
    0 \leq \Tilde{a}^{\s}_{0} \leq b_0 \leq \Tilde{a}^{\t}_{1} \leq b_1 < \pi
\end{equation}
for all $\s,\t \in\{-1,1\}$, i.e.~when then modified angles on Alice's side $\Tilde{a}^{\a}_x$ and the angles on Bob's side $b_y$ alternate for all choices of $(\s,\t)$. Conversely, one can conclude that when this full alternating property is not verified, there exists a solution to the problem \cref{eq:decomp} and as such the point $\Vec P$ is non-exposed in $\mathcal{Q}$. 

The cases $\theta\in [\pi/4,\pi/2)$, $\theta\in (\pi/2,3\pi/4]$ and $\theta\in [3\pi/4,\pi)$ were left aside, but the proof goes exactly as the previous case, but considering the solutions of all three linear systems for $(\s,\t) \neq (1,-1)$, $(\s,\t) \neq (1,-1)$ and $(\s,\t) \neq (-1,1)$ respectively. 

\end{document}